\documentclass[9pt, letterpaper, twocolumn, oneside]{article}
\pdfoutput=1

\usepackage[margin=2cm]{geometry}
\usepackage[OT1]{fontenc}
\usepackage{amsmath,amssymb,amsfonts}
\usepackage{amsthm}  
\usepackage{graphicx}
\usepackage{booktabs}  
\usepackage{multirow}  
\usepackage[hidelinks, bookmarks=true]{hyperref}  
\usepackage[noend]{algpseudocode}  
\usepackage[capitalize,nameinlink]{cleveref}  
\usepackage{mathtools}  
\usepackage[defaultlines=2,all]{nowidow}  
\usepackage{url}  
\usepackage[labelfont=bf]{caption}
\usepackage[caption=false,font=footnotesize]{subfig}  
\usepackage{tikz}
\usepackage{pgfplots}
\usepackage{siunitx}
\usepackage{mathdots}  
\usepackage{setspace}  
\usepackage[numbers, sort&compress]{natbib}  
\usepackage{todonotes}
\usepackage{adjustbox}  
\usepackage{colortbl}  

\usetikzlibrary{
    shapes.misc,
    shapes.geometric,
    matrix,
    decorations.pathreplacing,
    patterns,
    intersections,
    patterns,
    positioning, backgrounds,  
    fit,
}  

\usepgfplotslibrary{groupplots}

\newcommand{\algoBuild}{\textsf{Build\-\histname}}
\newcommand{\algoRefineOneDim}{\textsf{Refine\-Bin1D}}
\newcommand{\algoRefineTwoDim}{\textsf{Refine\-Bin2D}}
\newcommand{\algoIsUniform}{\textsf{IsUniform}}

\newcommand{\histfunc}{\textsf{Hist}}

\newcommand{\ndims}{d}
\newcommand{\nrows}{N}
\newcommand{\minpts}{M}
\newcommand{\testsignificance}{\alpha}
\newcommand{\nsamples}{N_s}
\newcommand{\sr}{\rho}
\newcommand{\edge}{e}
\newcommand{\edges}{\boldsymbol{\edge}}
\newcommand{\edgeinit}{\hat{\edge}}
\newcommand{\edgesinit}{\boldsymbol{\edgeinit}}
\newcommand{\edgesfromgd}{E_{0}}
\newcommand{\edgesrefined}[2]{\edges^{*{#2}}_{#1}}
\newcommand{\bincounts}{\boldsymbol{H}}
\newcommand{\nbins}{k}
\newcommand{\bincount}{h}
\newcommand{\binuniquecount}{u}
\newcommand{\binuniquecounts}{\boldsymbol{u}}
\newcommand{\binminmax}{v}
\newcommand{\binmin}[2]{\binminmax^{{#2}-}_{#1}}
\newcommand{\binmax}[2]{\binminmax^{{#2}+}_{#1}}
\newcommand{\binmins}[2]{\boldsymbol{\binminmax}^{{#2}-}_{#1}}
\newcommand{\binmaxs}[2]{\boldsymbol{\binminmax}^{{#2}+}_{#1}}
\newcommand{\binmidpoint}{c}
\newcommand{\binmidpoints}{\boldsymbol{c}}
\newcommand{\bincentreboundlow}[2]{c^{-{#2}}_{#1}}
\newcommand{\bincentreboundslow}[2]{\boldsymbol{c}^{-{#2}}_{#1}}

\newcommand{\bincentreboundshigh}[2]{\boldsymbol{c}^{+{#2}}_{#1}}

\newcommand{\dataset}{\mathcal{D}}
\newcommand{\datasetsample}{\boldsymbol{D}}
\newcommand{\datasetattribute}{X}
\newcommand{\datacolumn}{\boldsymbol{x}}
\newcommand{\datacolumns}{\boldsymbol{X}}
\newcommand{\datasetunique}{U}
\newcommand{\nunique}{N_{U}}
\newcommand{\datasetuniques}{\boldsymbol{\datasetunique}}
\newcommand{\binwidth}{\Delta}
\newcommand{\subbinwidth}{\delta}
\newcommand{\nsubbins}{s}
\newcommand{\splitpoint}{z}

\newcommand{\chisq}[1]{\chi^2_{#1}}
\newcommand{\subbincount}{\hslash}
\newcommand{\subbincountexpected}{\hat{\subbincount}}
\newcommand{\subbincountdelta}{\epsilon}
\newcommand{\binidx}{t}
\newcommand{\subbinidx}{r}

\newcommand{\datatypequanta}{\mu}  

\newcommand{\storage}{S}
\newcommand{\bytesrequired}{m}
\newcommand{\lencount}{\ell_{\bincount}}
\newcommand{\countstorageindicator}{\mathcal{I}_{\bincount}}
\newcommand{\nbinsnonzero}{\theta}
\newcommand{\golomb}{\mathcal{G}}

\newcommand{\coverage}{\beta}
\newcommand{\coverages}{\boldsymbol{\coverage}}
\newcommand{\weighting}{w}
\newcommand{\weightings}{\boldsymbol{\weighting}}
\newcommand{\weightingplaceholder}{\weighting^{\bullet}}
\newcommand{\weightingsplaceholder}{\weightings^{\bullet}}
\newcommand{\binkey}{t^{*}}
\newcommand{\fractomedian}{f_{\binkey}}
\newcommand{\nsubbinsfullycovered}{a}
\newcommand{\nsubbinspartiallycovered}{b}
\newcommand{\varboundsx}{\xi}
\newcommand{\varboundsxs}{\boldsymbol{\varboundsx}}
\newcommand{\prob}{\mathsf{Pr}}
\newcommand{\predicate}{P}
\newcommand{\predicates}{\boldsymbol{\predicate}}
\newcommand{\npredicates}{n}
\newcommand{\predicateidx}{\ell}
\newcommand{\colvar}{x}
\newcommand{\partialbincount}[2]{h_{#1|#2}}
\newcommand{\nsubbinscovered}{\bar{\nsubbins}}
\newcommand{\ciweight}{z_{0.98}}

\newcommand{\histname}{Pairwise\-Hist}
\newcommand{\low}{\scalebox{0.8}{$L$}}
\newcommand{\high}{\scalebox{0.8}{$R$}}

\newcommand{\myscale}[1]{\scalebox{0.8}{#1}}
\newcommand{\tablecaptiongapskrinkage}{-0.6em}

\DeclarePairedDelimiter\ceil{\lceil}{\rceil}
\DeclarePairedDelimiter\floor{\lfloor}{\rfloor}
\DeclarePairedDelimiter\norm{\lVert}{\rVert}
\newcommand{\lonenorm}[1]{\norm{#1}_{1}}

\newcommand{\setsep}{\,\big|\,}
\newcommand{\dotprod}{\cdot}
\newcommand{\hadamardproduct}{\odot}
\newcommand{\hadamarddivision}{\oslash}

\newcommand{\datacdfdbestpp}{data/error_percentiles_dbestpp_queries.csv}
\newcommand{\datacdfdeepdb}{data/error_percentiles_deepdb_queries.csv}
\newcommand{\datacdfall}{data/error_percentiles_all_queries.csv}
\newcommand{\datarealdatasetsacccuracy}{data/real_world_datasets_accuracy.csv}
\newcommand{\datarealdatasetsstorage}{data/real_world_datasets_storage.csv}
\newcommand{\dataForPairwiseHistFigure}{data/pairwisehist_figure_data.csv}
\newcommand{\dataForPairwiseHistFigureOne}{data/pairwisehist_figure_data_2d_0_1.csv}
\newcommand{\dataForPairwiseHistFigureTwo}{data/pairwisehist_figure_data_2d_0_2.csv}
\newcommand{\dataForPairwiseHistFigureThree}{data/pairwisehist_figure_data_2d_0_3.csv}
\newcommand{\dataForPairwiseHistFigureFour}{data/pairwisehist_figure_data_2d_1_2.csv}

\newtheorem{theorem}{Theorem}

\definecolor{colourAlgoComment}{HTML}{787878}

\definecolor{colour1}{HTML}{006ed4}  
\definecolor{colour2}{HTML}{e41a1c}  
\definecolor{colour3}{HTML}{4daf4a}  
\definecolor{colour4}{HTML}{984ea3}  
\definecolor{colour5}{HTML}{ff7f00}
\definecolor{colour6}{HTML}{a65628}
\definecolor{colour7}{HTML}{f781bf}
\definecolor{colour8}{HTML}{999999}

\renewcommand{\arraystretch}{1.2}
\input{algorithm_customisations}
\newcommand{\algowidth}{1.3\columnwidth}
\newcommand{\algocomment}[1]{\hfill\textcolor{colourAlgoComment}{// \textit{#1}}}  
\algnewcommand{\LineComment}[1]{\textcolor{colourAlgoComment}{// \textit{#1}}}  
\algnewcommand{\IfThen}[2]{
    \State \algorithmicif\ #1\ \algorithmicthen\ #2}
\newcounter{algoCounter}  
\newcommand{\showAlgoCounter}[2]{\noalign{\refstepcounter{algoCounter}#1}\textbf{Algorithm~\thealgoCounter} \textsf{#2}}  
\newcommand{\algolinespacing}{1.12}

\tikzset{
	every pin/.style={
		pin edge={black,thick},
		font=\footnotesize
	},
}

\newcommand{\graphStandardWidth}{0.75\columnwidth}
\newcommand{\graphStandardHeight}{0.65\columnwidth}

\pgfplotsset{
	compat=1.17,
	every axis/.style={
        axis on top,
		font=\footnotesize,
		table/col sep=comma,
        xtick pos=bottom,
        ytick pos=left,
		log origin=infty,  
        legend cell align={left},
        legend style={
            legend pos=north west,
            align=left,
            row sep=0em,
            column sep=0.2em,
            font=\tiny,
            fill opacity=0.75,
            text opacity=1,
            /tikz/nodes={inner sep=0.1em},
            /tikz/every odd column/.style={yshift=0.1em},
        },
		title style={yshift=-0.6em},
        scale only axis,  
        width=\graphStandardWidth, height=\graphStandardHeight,
        xticklabel style={text height=0.5em},  
	},
    my bar plot/.style={
        width  = 0.42*\columnwidth,
        height = 0.13*\textwidth,
        xtick = data,
        scaled y ticks = false,
        enlarge x limits=0.25,
        bar width=6.5pt,
        ybar=0,
        ymin=0,
        area legend,
        legend image post style={scale=0.6},
        cycle list={
            {white, fill=colour3!70},
            {white, fill=colour3!70, postaction={pattern=north east lines, pattern color=white}},
            {white, fill=colour1!70},
            {white, fill=colour1!70, postaction={pattern=north east lines, pattern color=white}},
            {white, fill=colour2!70},
            {white, fill=colour2!70, postaction={pattern=north east lines, pattern color=white}},
            {white, fill=colour4!70},
            {white, fill=colour4!70, postaction={pattern=north east lines, pattern color=white}},
            {white, fill=colour5!70},
            {white, fill=colour5!70, postaction={pattern=north east lines, pattern color=white}},
        }
    },
    parameter sensitivity plots/.style={
        width  = 0.42*\columnwidth,
        height = 0.33*\columnwidth,
        xtick = data,
        xlabel = {Min. points, $ \minpts $},
        ymin=0,
        symbolic x coords={1000, 2000, 3000, 4000, 5000, 6000, 7000, 8000, 9000, 10000},
        xtick={1000,4000,7000,10000},
        xticklabels={1\,000, 4\,000, 7\,000, 10\,000},
        legend image post style={scale=0.8},
        cycle list={
            {colour1!70, mark options={draw=white,thick,scale=1.12,opacity=1}, mark=*},
            {colour2!70, mark options={draw=white,thick,scale=1.0,opacity=1}, mark=square*},
        },
        line width=0.5pt,
    },
    evaluation bar plots/.style={
        my bar plot,
        width = 0.42*\columnwidth, height = 0.31*\columnwidth,
        enlarge x limits=0.6,
        xtick = {1,2},
        xticklabels = {Power, Flights},
    },
    dataset grouped bar plot/.style={
        my bar plot,
        scale only axis=false,
        width=1.6\columnwidth,
        height=0.5\columnwidth,
        symbolic x coords={Aqua, Basement, Build, Current, Flights, Furnace, Gas, Light, Power, Taxis, Temp},
        x tick label style={rotate=30, anchor=north east, inner sep=0.3em, xshift=0.2em},
        log ticks with fixed point,
        minor tick style={draw=none},
        enlarge x limits=0.075,
        bar width=3.5pt,
        transpose legend=true,
        legend columns=2,
        legend style={
            legend pos=north east,
        },
    },
    cdf plots/.style={
        width = 0.19*\textwidth, height = 0.16*\textwidth,
        ymin=0, ymax=100,
        xmin=1e-3, xmax=1e2,
        ylabel={Percentile},
        xlabel={Relative error},
        xtick={1e-3, 1e-2, 1e-1, 1e-0, 1e1, 1e2},
        xtick align=outside,
        minor tick style={draw=none},
        legend columns=1,
        legend image post style={scale=0.4},
        legend style={
            align=left,
            row sep=0em,
            font=\tiny,
            /tikz/nodes={inner sep=0.1em},
            /tikz/every even column/.append style={column sep=0.3em},
            /tikz/every odd column/.append style={column sep=0.1em, fill opacity=1},
            anchor=north east,
            legend pos=north west,
            fill opacity=0.75,
            text opacity=1,
        },
        cycle list={
            {colour3!80, thick},
            {colour3!80, thick, densely dashed},
            {colour1!80, thick},
            {colour1!80, thick, densely dashed},
            {colour2!80, thick},
            {colour2!80, thick, densely dashed},
        },
        line width=0.5pt,
    },
	/pgf/declare function={Floor(\x) = round(\x-0.49);},  
	custom line plot style/.style={
		cycle list={
			{colour1!80, thick, mark options={draw=white,thick,opacity=1,scale=1.3}, mark=*},
			{colour2!50, thick, mark options={draw=white,thick,opacity=1,scale=1.15}, mark=square*},
			{colour2!90, thick, mark options={draw=white,thick,opacity=1,scale=1.7}, mark=triangle*},
			{colour3!50, thick, mark options={draw=white,thick,opacity=1,scale=1.4}, mark=pentagon*},
			{colour3!100, thick, mark options={draw=white,thick,opacity=1,scale=1.8}, mark=diamond*},
		},
		line width=0.5pt,
		legend style={
			align=left,
			column sep=0.1em,
			/tikz/every odd column/.style={yshift=0.1em},
			/tikz/nodes={inner sep=0.1em},
		},
	},
	custom scatter plot style/.style={
		only marks,
		cycle list={
			{colour1,mark options={scale=1.0},mark=*},
			{colour2,mark options={scale=1.0},mark=*},
			{colour3,mark options={scale=1.0},mark=*},
			{colour4,mark options={scale=1.0},mark=*},
			{colour5,mark options={scale=1.0},mark=*},
			{colour6,mark options={scale=1.0},mark=*},
			{colour7,mark options={scale=1.0},mark=*},
			{colour8,mark options={scale=1.0},mark=*}
		},
	}
}

\begin{document}

	\title{\histname{}: Fast, Accurate and Space-Efficient Approximate Query Processing with Data Compression}

    \author{
        Aaron Hurst, Daniel E. Lucani and Qi Zhang
        \\
        \footnotesize DIGIT \& Department of Electrical and Computer Engineering \\[-0.2cm]
        \footnotesize Aarhus University, Denmark\\[-0.2cm]
        \footnotesize \{ah, daniel.lucani, qz\}@ece.au.dk
    }

	\maketitle

	\begin{abstract}
    Exponential growth in data collection is creating significant challenges for data storage and analytics latency.
    Approximate Query Processing (AQP) has long been touted as a solution for accelerating analytics on large datasets, however, there is still room for improvement across all key performance criteria.
    In this paper, we propose a novel histogram-based data synopsis called \histname{} that uses recursive hypothesis testing to ensure accurate histograms and can be built on top of data compressed using Generalized Deduplication (GD).
    We thus show that GD data compression can contribute to AQP.
    Compared to state-of-the-art AQP approaches, \histname{} achieves better performance across all key metrics, including 2.6$ \times $ higher accuracy, 3.5$ \times $ lower latency, 24$ \times $ smaller synopses and 1.5--4$ \times $ faster construction time.
\end{abstract}

    \section{Introduction}
\label{sec:introduction}

Rapidly advancing digital transformation is driving exponential growth in data volumes across many sectors.
This poses significant challenges for current infrastructure and data management solutions, which must handle swelling data workloads increasingly advanced analytics.
Efficient data storage and analytics techniques are therefore crucial.


Approximate Query Processing (AQP) is a well-established field that focuses on enabling fast analytics on Big Data by sacrificing some degree of accuracy~\cite{Li_2018, Akash_2022}.
AQP techniques are typically based on either sampling or data synopses, or a hybrid of the two~\cite{Liang_2021,Li_2018,Peng_2018,Park_2018}.
In general, small samples or compact synopses enable analytics to be performed over large datasets within required latency constraints.
However, all AQP approaches exhibit a distinct trade-off between accuracy, latency and synopsis size.


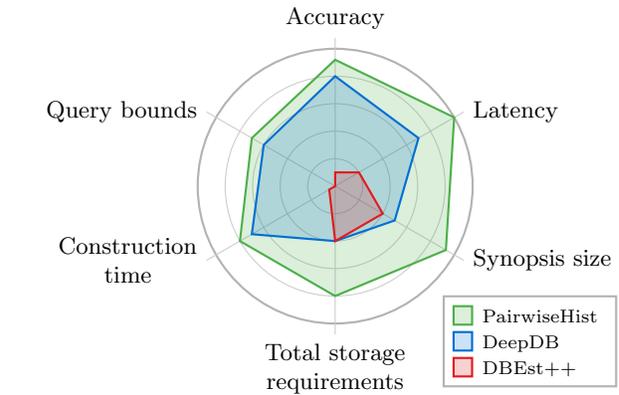
\begin{figure}[!t]
    \centering
    \newcommand{\radarradius}{5.2em}

\begin{tikzpicture}[
    grid/.style={black!20},
    metric label/.style={font=\small, black, anchor=west, align=center},
    legend item/.style={
        minimum height=0.7em,
        minimum width=0.7em,
        thick,
        draw,
        fill opacity=0.2,
        anchor=north,
        inner sep=0,
    }
    ]

    \draw[grid, thick, black!30] (0,0) circle (\radarradius);
    \draw[grid] (0,0) circle (0.8*\radarradius);
    \draw[grid] (0,0) circle (0.6*\radarradius);
    \draw[grid] (0,0) circle (0.4*\radarradius);
    \draw[grid] (0,0) circle (0.2*\radarradius);

    \draw[grid] (0,0) -- ( 30:1.08*\radarradius) node[metric label] {Latency};
    \draw[grid] (0,0) -- ( 90:1.08*\radarradius) node[metric label, anchor=south] {Accuracy};
    \draw[grid] (0,0) -- (150:1.08*\radarradius) node[metric label, anchor=east] {Query bounds};
    \draw[grid] (0,0) -- (210:1.08*\radarradius) node[metric label, anchor=east] {Construction\\time};
    \draw[grid] (0,0) -- (270:1.08*\radarradius) node[metric label, anchor=north] {Total storage\\requirements};
    \draw[grid] (0,0) -- (330:1.08*\radarradius) node[metric label] {Synopsis size};


    \draw[colour3, thick, fill=colour3, fill opacity=0.2]
    ( 30:1.0*\radarradius) 
    -- ( 90:0.92*\radarradius) 
    -- (150:0.7*\radarradius) 
    -- (210:0.8*\radarradius) 
    -- (270:0.8*\radarradius) 
    -- (330:0.93*\radarradius) 
    -- cycle;

    \draw[colour1, thick, fill=colour1, fill opacity=0.2]
    ( 30:0.7*\radarradius)  
    -- ( 90:0.8*\radarradius)  
    -- (150:0.6*\radarradius)  
    -- (210:0.7*\radarradius)  
    -- (270:0.4*\radarradius)  
    -- (330:0.5*\radarradius)  
    -- cycle;

    \draw[colour2, thick, fill=colour2, fill opacity=0.2]
    ( 30:0.2*\radarradius)  
    -- ( 90:0.1*\radarradius)  
    -- (150:0.0*\radarradius)  
    -- (210:0.05*\radarradius)  
    -- (270:0.4*\radarradius)  
    -- (330:0.4*\radarradius)  
    -- cycle;


    \node[legend item, colour3, fill=colour3] (legend1) at (317:1.27*\radarradius) {};
    \node[legend item, colour1, fill=colour1] (legend2) at ([yshift=-0.2em]legend1.south) {};
    \node[legend item, colour2, fill=colour2] (legend3) at ([yshift=-0.2em]legend2.south) {};

    \node[font=\scriptsize, anchor=west, text depth=0.1em, text height=0.6em] (label1) at (legend1.east) {\histname};
    \node[font=\scriptsize, anchor=west, text depth=0.1em, text height=0.6em] (label2) at (legend2.east) {DeepDB};
    \node[font=\scriptsize, anchor=west, text depth=0.1em, text height=0.6em] (label3) at (legend3.east) {DBEst++};

    \node[fit=(legend1)(label1.east)(legend3), draw, thick, black!30] () {};

\end{tikzpicture}
    \vspace{-0.2em}
    \caption{
        Relative performance comparison of \histname{}, DeepDB and DBEst++ summarising \cref{fig:real_datasets_results,fig:error,fig:evaluation_other} and \autoref{tab:bounds}.
        Outer rings indicate better performance.
        Each interval represents approximately 2$ \times $ improvement.
    }
    \label{fig:performance_radar}
\end{figure}


In this paper, we propose a novel histogram-based data synopsis for AQP called \histname{} that consists of three key components:
i)~one-dimensional histograms that capture within-column data distributions,
ii)~two-dimensional histograms that capture relationships between each pair of columns (hence, \textit{Pairwise}Hist), and
iii)~small metadata for each histogram bin that enhance query precision, including the minimum, maximum and number of unique values.
By using histograms, \histname{} inherits desirable properties including accurate query bounds and effective outlier recall, while the small number of histograms minimises synopsis size.
All histograms within \histname{} are constructed using recursive hypothesis testing that ensures each bin contains uniformly distributed data, leading to high query accuracy.
Low query execution latency is realised by novel methods for resolving multi-predicate queries that require only a few relatively small matrix multiplications.
Overall, \histname{} delivers all-round superior performance across accuracy, latency, synopsis size, synopsis construction time and query bounds compared to state-of-the-art AQP techniques, as illustrated in \autoref{fig:performance_radar}.


\histname{} takes inspiration from recent works in data compression that demonstrate (approximate) data clustering can be performed directly on compressed data without decompression~\cite{Hurst_2021,Hurst_2022,Hurst_2024}.
That is, by using Generalized Deduplication (GD) data compression~\citep{Vestergaard_2019a,Vestergaard_2019b,Vestergaard_2020,Vestergaard_2021}, part of the compressed data known as \textit{bases} can serve as a data synopsis on which analytics tasks can be performed efficiently.
Compared to GD bases, \histname{} significantly reduces storage requirements and improves accuracy by using low-dimensional histograms with variable precision.


While \histname{} is a stand-alone AQP technique in its own right, we also propose implementing it alongside data compression, as illustrated in \autoref{fig:system}.
Specifically, we use GreedyGD~\cite{Hurst_2024}, a recent version of GD, which is a lossless data compression algorithm that offers state-of-the-art compression ratios and low random access cost.
While originally designed for IoT, GD has also proved effective in several other domains~\citep{Hurst_2024, Sehat_2022,Nielsen_2019,Feher_2022,Goettel_2020}.
In our proposed AQP framework, GreedyGD both reduces overall storage requirements and accelerates synopsis construction.


\begin{figure*}[!t]
    \centering
    \input{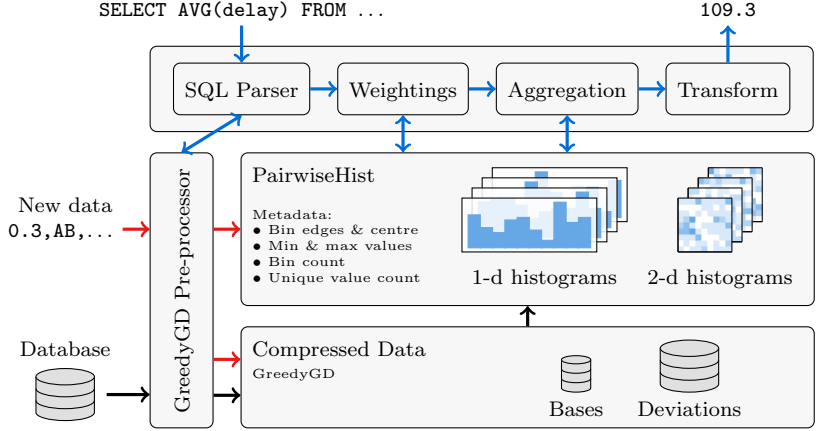}
    \vspace{-0.2em}
    \caption{
        Our proposed AQP framework with compression, including data ingestion and \histname{} construction (black arrows), query execution (blue) and data updates (red).
    }
    \label{fig:system}
    \vspace{0em}
\end{figure*}


Due to its small synopsis size and low query latency, our approach also enables Edge analytics, even on resource-constrained devices.
This comes with many benefits, including low (communication) latency, scalability, privacy, energy savings and mobility~\cite{Zhang_2014,Kong_2022,Nayak_2022}.
Meanwhile, \histname{}'s fast construction and low storage requirements reduce Cloud storage and computing costs and enable more frequent updates.


In summary, the key contributions of this paper include:
\begin{enumerate}
    \item A novel histogram-based AQP technique called \histname{} that uses a combination of one- and two-dimensional histograms and an efficient storage encoding to deliver data synopses that are 24$ \times $ smaller and up to 4$ \times $ faster to construct than state-of-the-art AQP methods DeepDB~\cite{Hilprecht_2020} and DBEst++~\cite{Ma_2021},

    \item Query execution techniques for seven common aggregation functions that deliver at least 2.6$ \times $ better accuracy (0.28\% median error vs. 0.73--28.9\%), 3.5$ \times $ lower latency and more accurate bounds than the state-of-the-art,

    \item A novel AQP framework that integrates data compression to reduce overall storage requirements by up to 4.3$ \times $, and

    \item A comprehensive performance evaluation of \histname{} in comparison to the state-of-the-art across 11 real-world datasets and two datasets scaled-up using IDEBench~\cite{Eichmann_2020}.
\end{enumerate}


The remainder of this paper is structured as follows.
\autoref{sec:related_work} describes related work.
\autoref{sec:system} provides our system overview.
\autoref{sec:histogram} outlines the \histname{} data structure and construction algorithm.
\autoref{sec:query_execution} defines the theory and mathematical formulations for query execution.
\autoref{sec:evaluation} evaluates the performance of \histname{}.
Finally, \autoref{sec:conclustion} concludes the paper.


	\section{Related Work}
\label{sec:related_work}

AQP methods are typically classified as either sampling-based or synopses-based~\citep{Akash_2022}.
Sampling approaches can be either offline or online~\cite{Chaudhuri_2017}.
Online methods select samples at query time and attempt to minimise the amount of data that must be accessed to achieve a desired accuracy.
Gapprox~\citep{Ahmadvand_2019}, for example, uses cluster sampling to reduce processing. 
Offline approaches, on the other hand, attempt to prepare in advance the most representative sample for the most queries.
For example, \citet{Babcock_2003} use biased sampling and BlinkDB~\citep{Agarwal_2013} uses stratified sampling.
More recent offline sampling work has focused on additional features, such as integration middleware~\citep{Park_2018} and increasing sample re-usability~\citep{Sanca_2023}.
However, sampling methods generally struggle with skewed data and may support limited aggregation functions~\cite{Li_2018}.

Synopsis-based approaches focus on building a compact summary of the data using statistical or machine learning techniques.
Histograms are a classical synopsis approach that is widely used in selectivity estimation, which involves estimating the number of tuples that a query will access.
This is an important step in database query optimisers and is equivalent to AQP for COUNT queries.
Many histogram algorithms exist, including simple equi-width and equi-depth histograms, as well as advanced methods, such as V-optimal histograms~\citep{Ioannidis_1995}, entropy-based histograms~\citep{To_2013} and others~\cite{Acharya_2015,Diakonikolas_2018,Singla_2022}.
Most approaches, however, are limited to one dimension.
Multi-dimensional histograms are notoriously challenging to construct and their storage scales exponentially with the number of dimensions~\citep{Cormode_2011,Zhang_2021a}.
Nonetheless, some multi-dimensional histograms for selectivity estimation are available, including DigitHist~\citep{Shekelyan_2017}, DMMH~\citep{Zhang_2021a} and STHoles~\citep{Bruno_2001}.
To avoid the pitfalls of high dimensionality, DigitHist uses lossy compression, DMMH uses density modelling and STHoles uses previous query results.
\citet{Cormode_2011} also discuss using collections of histograms, which is the approach we have taken with \histname{}.
Common to all histograms is that accuracy depends on ensuring a uniform distribution of tuples within each bin.
To the best of our knowledge, our approach is the first to utilise hypothesis testing for this purpose.

Most recent synopsis-based AQP works focus on machine learning.
These can be classified as either generative, which create synthetic data for evaluating queries, or inferential, which directly predict AQP results~\cite{Lee_2022}.
In \citep{Kulessa_2018,Kulessa_2019}, a generative model that captures the joint probability distribution of a database using Mixed Sum-Product Networks (MSPNs) is proposed.
This approach can evaluate simple queries directly using the MSPN weights or generate (synthetic) samples to answer more complex queries.
While this provides high accuracy, low latency and compact summaries, it requires significant construction time (hours for just $ 10^6 $ samples).

DeepDB~\cite{Hilprecht_2020} proposes Relational Sum Product Networks (RSPNs), which are purely inferential and extend (M)SPNs to support multiple database tables, complex queries and model updates.
In comparison to \citep{Kulessa_2018}, DeepDB achieves much faster construction, but has poor latency with multi-predicate queries, higher storage requirements and can give imprecise bounds.
Our evaluation also revealed that DeepDB does not support OR relationships between predicates, despite claiming to, and performs poorly on real-world data.

DBEst~\cite{Ma_2019} is an inferential approach that uses a combination of model types.
Kernel density estimators model individual columns, while regression models capture relationships between pairs of columns.
This is similar to \histname{} with the density and regression models corresponding to one- and two-dimensional histograms, respectively.
A significant limitation of DBEst is that a new model is required for every query template, which limits flexibility and exponentially increases synopsis storage requirements.

DBEst++~\cite{Ma_2021} improves DBEst by switching to mixture density networks for both regression and density modelling.
This significantly reduces storage requirements, delivers better accuracy and latency, and supports data updates.
However, storage requirements for DBEst++ are misleading due to each model template requiring its own model.
Our tests also revealed that DBEst++ does not support queries involving more than two columns, OR relationships between predicates, queries on only categorical columns or inequality predicates on date/time columns.

Other machine learning approaches include LAQP~\citep{Zhang_2021}, which combines an error prediction model with sampling, Electra~\citep{Sheoran_2022}, which focuses on queries with many-predicates, NeuroSketch~\citep{Zeighami_2023}, which is a bounded inferential model trained on queries, and Generative Adversarial Networks~\citep{Fallahian_2022}.
Typically, machine learning methods are constructed from samples and exhibit similar limitations as sampling-based AQP methods, namely vulnerability to skewed data and limited aggregation function support.

A small number of unified approaches that combine online sampling and synopses have also been proposed.
For example, AQP++~\cite{Peng_2018} generates a set of pre-computed aggregations known as a prefix cube and answers queries by supplementing relevant prefix cube elements with sampling.
A similar approach has also been proposed for selectivity estimation~\cite{M_ller_2018}.
More recently, \citep{Liang_2021} proposed PASS, which has a more flexible synopsis design, reminiscent of  DigitHist~\cite{Shekelyan_2017}, that consists of a hierarchical set of pre-computed aggregates at different resolutions that guide online stratified sampling.
Unlike other sampling-based approaches, PASS provides both probabilistic and deterministic bounds.
However, a significant limitation is nearly 30$ \times $ slower construction than AQP++~\cite{Liang_2021}, which itself requires over 20~minutes for just a \qty{50}{MB} sample~\cite{Peng_2018}.

Many AQP approaches provide query error bounds, which give analysts an indication of the confidence that they can place in AQP results.
Bounds are typically based on probabilistic confidence intervals~\cite{Li_2018,Zeng_2014}, but can also be deterministic~\citep{Agarwal_2013,Li_2019,Ahmadvand_2019}.
Unfortunately, probabilistic bounds can often be incorrect~\citep{Agarwal_2014} while deterministic bounds may be too broad to be useful.

\begin{table*}[!t]
    \centering
    \caption{
        \histname{} compared to previous AQP works.
    }
    \footnotesize
    \newcommand*{\belowrulesepcolor}[1]{%
    \noalign{%
        \kern-\belowrulesep
        \begingroup
        \color{#1}%
        \vspace{0.1\lightrulewidth}
        \hrule height \belowrulesep
        \endgroup
    }%
}
\newcommand*{\aboverulesepcolor}[1]{%
    \noalign{%
        \begingroup
        \color{#1}%
        \hrule height\aboverulesep
        \endgroup
        \kern-\aboverulesep
    }%
}

\setlength\tabcolsep{2.5pt}

\begin{tabular}{lcccccc}
    \toprule
                                                                             & \textbf{Accuracy} & \textbf{Latency} & \textbf{Bounds} &  \textbf{Size}  & \textbf{Build} & \textbf{Versatility} \\ \midrule
    \belowrulesepcolor{colour3!15}\rowcolor{colour3!15} \textbf{\histname{}} & \textbf{$ < $1\%} & \textbf{sub-ms}  &  \textbf{yes}   & \textbf{sub-MB} & \textbf{secs}  &  \textbf{very high}  \\
    \aboverulesepcolor{colour3!15} \midrule
    VerdictDB~\citep{Park_2018}  &        1\%        &     seconds      &  \textbf{yes}   &       GBs       &       ?        &  \textbf{very high}  \\
    Gapprox~\citep{Ahmadvand_2019}                                           &     $ < $5\%      &     seconds      &  \textbf{yes}   &       n/a       &      n/a       &         low          \\
    BlinkDB~\citep{Agarwal_2013}                                             &    $ < $10 \%     &     seconds      &  \textbf{yes}   &       GBs       &      n/a       &         high         \\ \midrule
    DigitHist~\citep{Shekelyan_2017}                                         &        1\%        & \textbf{sub-ms}  &  \textbf{yes}   &       MBs       &      mins      &       very low       \\
    DMMH~\citep{Zhang_2021a}                                                 &      1--2\%       &        ms        &       no        & \textbf{sub-MB} & \textbf{secs}  &       very low       \\
    STHoles~\citep{Bruno_2001}                                               &       10\%        &        ?         &       no        & \textbf{sub-MB} &       ?        &       very low       \\ \midrule
    DeepDB~\citep{Hilprecht_2020}                                            &        1\%        &        ms        &  \textbf{yes}   &       MBs       &      mins      &         high         \\
    DBEst++~\cite{Ma_2021}                                                   &    1\%$ ^{*} $    &        ms        &       no        &       MBs       &     hours      &         low          \\
    NeuroSketch~\citep{Zeighami_2023}                                        &        5\%        & \textbf{sub-ms}  &  \textbf{yes}   & \textbf{sub-MB} &      mins      &  \textbf{very high}  \\
    LAQP~\citep{Zhang_2021}                                                  &       10\%        &        ms        &       no        & \textbf{sub-MB} &       ?        &  \textbf{very high}  \\
    Electra~\citep{Sheoran_2022}                                             &       10\%        &        ?         &       no        &        ?        &       ?        &         low          \\ \midrule
    PASS~\citep{Liang_2021}                                                  & \textbf{$ < $1\%} &        ms        &  \textbf{yes}   &       MBs       &      mins      &         high         \\
    AQP++~\citep{Peng_2018}                                                  & \textbf{$ < $1\%} &     seconds      &  \textbf{yes}   &       MBs       &      mins      &         high         \\ \bottomrule
    \multicolumn{6}{l}{\scriptsize ``?'' indicates not reported by the authors. $ ^* $Much larger error observed in practice.}                                           &
\end{tabular}



    \label{tab:stateoftheart}
\end{table*}

The overall performance of state-of-the-art AQP techniques is summarised in \autoref{tab:stateoftheart}, in which versatility refers to the variety of supported query templates.
As can be seen, \histname{} delivers comprehensively superior performance.
Moreover, by using data compression, it uniquely offers significant overall storage reduction.

    \section{System Overview}
\label{sec:system}


\textbf{\textit{Problem Definition.}}
Consider a dataset $ \dataset $ with $ \nrows $ rows and $ \ndims $ attributes $ X_1, \ldots, \datasetattribute_{\ndims} $ and queries of the form:
\begin{align*}
    &\text{SELECT } F (\datasetattribute_{i}) \text{ FROM } \dataset \text{ WHERE } P_1 \text{ AND/OR } P_2 \ldots \\[-0.2em]
    &\text{GROUP BY } \ldots ;
\end{align*}
where $ F $ is an aggregation function (e.g. AVG), $ P_1, P_2, \ldots $ are predicate conditions of the form ``$ \datasetattribute_j \ OP \text{ LITERAL} $'', where $ OP $ is a binary logical operator (i.e., $ <, >, \le, \ge, = \text{ or } \ne $) and LITERAL is a valid value for column $ \datasetattribute_j $, and GROUP BY can be applied to any categorical column.
It is assumed that $ \dataset $ is large enough such that exact query execution is prohibitively expensive.
Therefore, the task is to design a framework such that bounded approximate query results can be obtained with high accuracy and low latency, while minimising synopsis size, synopsis construction time and overall storage requirements.
Missing values must also be supported.


\textbf{\textit{Data Compression.}}
In our proposed AQP framework, \histname{} is applied on top of compressed data, which reduces storage requirements.
As shown in \autoref{fig:system}, GreedyGD compresses incoming data, which includes pre-processing to improve compressability.
Pre-processing is applied to each column independently based on its data type and includes minimum value subtraction, floating point to integer conversion (e.g. 10.22 to 1022), frequency-ranked categorical value encoding (i.e., most common encoded as 0, second most as 1, etc.) and encoding missing values.
Importantly, pre-processing does not require additional storage or memory and datasets can be processed in arbitrarily-sized batches, which allows processing of datasets that are too large to fit in memory.


\begin{figure*}[!t]
    \centering
    \footnotesize

\newcommand{\verticalspacing}{\pgflinewidth}

\tikzset{
    standard node/.style={
        draw,
        fill=white,
        yshift=\verticalspacing,
        anchor=north,
        minimum height=1.6em,
        text height=1.0em,
        inner sep=0em,
        text depth=0.3em
    },
    chunk style/.style={
        standard node,
        minimum width=8.5em,
        fill=black!5,
    },
    base style/.style={
        standard node,
        minimum width=5.5em,
        fill=colour1!15,
    },
    id style/.style={
        standard node,
        minimum width=1.5em,
        fill=white,
    },
    deviation style/.style={
        standard node,
        minimum width=3.0em,
        fill=black!5,
    }
}

\begin{tikzpicture}
    \matrix[label={[text depth=0.3em]above:\textbf{Data}}, inner sep=0] (data) {
        \node[chunk style, yshift=\verticalspacing] (c1) {chunk 1};
        \node[chunk style] (c2) at (c1.south) {chunk 2};
        \node[chunk style] (c3) at (c2.south) {chunk 3};
        \node[chunk style] (c4) at (c3.south) {$ \ldots $};
        \node[chunk style] (c5) at (c4.south) {chunk $ n $}; \\
    };

    \matrix[label={[text depth=0.3em]above:\textbf{Bases \& Deviations}}, inner sep=0, anchor=west, xshift=3.0em] (splits) at (data.east) {
        \node[base style, anchor=west, yshift=-\verticalspacing] (b1) {base 1};
        \node[base style] (b2) at (b1.south) {base 2};
        \node[base style] (b3) at (b2.south) {base 1};
        \node[base style] (b4) at (b3.south) {$ \ldots $};
        \node[base style] (b5) at (b4.south) {base 2};

        \node[deviation style, anchor=west, xshift=-\pgflinewidth, yshift=-\verticalspacing] (d1) at (b1.east) {dev 1};
        \node[deviation style] (d2) at (d1.south) {dev 2};
        \node[deviation style] (d3) at (d2.south) {dev 3};
        \node[deviation style] (d4) at (d3.south) {$ \ldots $};
        \node[deviation style] (d5) at (d4.south) {dev $ n $}; \\
    };

    \matrix[label={[text depth=0.3em]above:\textbf{Deduplicated}}, inner sep=0, anchor=west, xshift=3.0em] (deduplicated) at (splits.east) {
        \node[base style, anchor=west, yshift=-\verticalspacing] (b1d) {base 1};
        \node[base style] (b2d) at (b1d.south) {base 2};
        \node[base style] (b3d) at (b2d.south) {$ \ldots $};

        \node[id style, anchor=west, yshift=-\verticalspacing, xshift=1em] (id1) at (b1d.east) {[1]};
        \node[id style] (id2) at (id1.south) {[2]};
        \node[id style] (id3) at (id2.south) {[1]};
        \node[id style] (id4) at (id3.south) {$ \ldots $};
        \node[id style] (id5) at (id4.south) {[2]};

        \node[deviation style, anchor=west, xshift=-\pgflinewidth, yshift=-\verticalspacing] (d1d) at (id1.east) {dev 1};
        \node[deviation style] (d2d) at (d1d.south) {dev 2};
        \node[deviation style] (d3d) at (d2d.south) {dev 3};
        \node[deviation style] (d4d) at (d3d.south) {$ \ldots $};
        \node[deviation style] (d5d) at (d4d.south) {dev $ n $}; \\
    };

    \draw[->, very thick, black] ([xshift=0.4em]data.east) to ([xshift=-0.4em]splits.west);
    \draw[->, very thick, black] ([xshift=0.4em]splits.east) to ([xshift=-0.4em]deduplicated.west);

    \draw[densely dashed] (b1d.east) -- (id1.west);
    \draw[densely dashed] (b1d.east) -- (id3.west);
    \draw[densely dashed] (b2d.east) -- (id2.west);
    \draw[densely dashed] (b2d.east) -- (id5.west);

\end{tikzpicture}
    \caption{
        GD splits data into bases and deviations.
    }
    \label{fig:gd}
\end{figure*}
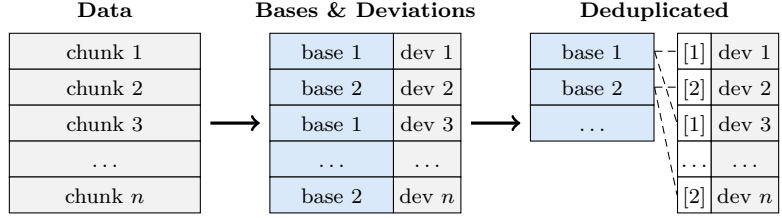

GreedyGD splits data \textit{chunks} into \textit{bases} and \textit{deviations}.
Bases contain the majority of the information and are deduplicated, while deviations are stored verbatim with IDs linking them to the appropriate bases, as illustrated in \autoref{fig:gd}.
Compression is achieved when there are few bases compared to the number of data chunks.
In our framework, a chunk corresponds to a row in a relational database table and bases contain the most significant bits from each attribute.
New rows can be added incrementally to the compressed data.


\textbf{\textit{\histname{}.}}
Once the data is compressed, \histname{} is built on top of the compressed data by taking (a sample of) the bases as input for the initial histogram bin edges.
The histograms are then refined using hypothesis testing to ensure that the distributions of tuples within individual bins are approximately uniform.
This process is described in detail in \autoref{sec:histogram}.
As illustrated in \autoref{fig:system}, \histname{} consists of one-dimensional histograms for each column, two-dimensional histograms for every pair of columns and various metadata for each histogram bin.
This collection of histograms enables rapid query aggregation on any column subject to arbitrary combinations of predicate conditions.
Furthermore, compared to a na\"ive multi-dimensional histogram, for which storage increases exponentially with the number of dimensions~\citep{Zhang_2021a}, our approach reduces storage requirements from $ O(\nbins^\ndims) $ to $ O(\ndims^2\nbins^2) $ for $ \ndims $ dimensions and $ \nbins $ bins per dimension.
A compact storage encoding for \histname{} is also proposed in \autoref{sec:histogram:storage}.

\histname{} can also support multi-table databases.
That is, queries across different tables can be resolved via two-dimensional histograms involving the primary/foreign keys.
However, in this paper we focus on single-table queries.

If \histname{} is used independently of GreedyGD, then histograms are constructed from scratch (i.e., without using bases as initial bin edges), with slightly longer synopsis construction time due to less precise initial conditions.

	\section{\histname{}}
\label{sec:histogram}

\begin{table}[!t]
    \centering
    \renewcommand{\arraystretch}{1.07}
    \caption{Notations}
    \footnotesize
    \begin{tabular}{ccl}
    \toprule
    \multirow{6}{*}{\rotatebox{90}{Construction}} &              $ \ndims $              & No. columns in the dataset                         \\
                                                    &              $ \nrows $              & No. rows in the dataset                            \\
                                                    &            $ \nsamples $             & No. samples used to construct \histname{}          \\
                                                    &               $ \sr $                & \histname{} sampling ratio, $ \nsamples / \nrows $ \\
                                                    &             $ \minpts $              & Minimum points for bins to be split                \\
                                                    &        $ \testsignificance $         & Significance level for hypothesis tests            \\ \midrule
     \multirow{7}{*}{\rotatebox{90}{Histograms}}  &           $ \nbins^{(i)} $           & No. histogram bins in column $ i $                 \\
                                                    &           $ \edges^{(i)} $           & Bin edges for column $ i $                         \\
                                                    &         $ \bincounts^{(i)} $         & Bin counts matrix for column $ i $                 \\
                                                    & $ \boldsymbol{\binminmax}^{(i)\pm} $ & Bin minimum and maximum values for column $ i $    \\
                                                    &       $ \binmidpoints^{(i)} $        & Bin midpoints for column $ i $                     \\
                                                    &      $ \binmidpoints^{(i)\pm} $      & Bin weighted centre bounds for column $ i $        \\
                                                    &      $ \binuniquecounts^{(i)} $      & No. unique values in each bin for column $ i $     \\ \midrule
     \multirow{4}{*}{\rotatebox{90}{Hypothesis}}  &            $ \binwidth $             & Bin width                                          \\
                                                    &           $ \subbinwidth $           & Sub-bin width                                      \\
                                                    &            $ \nsubbins $             & No. sub-bins in a given bin                        \\
                                                    &           $ \subbincount $           & Sub-bin count                                      \\ \midrule
      \multirow{5}{*}{\rotatebox{90}{Queries}}    &            $ \coverages $            & Bin coverage                                       \\
                                                    &           $ \weightings $            & Bin weightings                                     \\
                                                    &      $ \nsubbinsfullycovered $       & No. sub-bins fully covered by a query              \\
                                                    &    $ \nsubbinspartiallycovered $     & No. sub-bins fully or partially covered by a query \\
                                                    &             $ \binkey $              & Bin index of query result (MEDIAN/MIN/MAX)    \\ \bottomrule
\end{tabular}

    \label{tab:nomenclature}
\end{table}

\histname{} generates its collection of histograms by recursively refining (i.e., splitting) an initial set of histogram bins using hypothesis testing, which ensures that data points within individual bins are sufficiently uniformly distributed.
Refinement is terminated when a bin is either uniform or has too few points to be split further.
\histname{} is thus parameterised by the minimum number of points required for a bin to be split, $ \minpts $, the hypothesis test significance level,~$ \testsignificance $, and the number of samples used to construct \histname{}, $ \nsamples $.
Before describing \histname{} in detail, this section provides an overview of the key notation, which is summarised in \autoref{tab:nomenclature}.

Bin edges for one-dimensional histograms are denoted by the vector $ \edges^{(i)} = \left\langle \edge^{(i)}_{1}, \edge^{(i)}_{2}, \ldots, \edge^{(i)}_{\nbins^{(i)}} \right\rangle $, where the superscript $ (i) $ indicates the column to which the histogram applies and $ \nbins^{(i)} $ denotes the number of bins for the one-dimensional histogram for column $ i $.
For two-dimensional histograms, the bin edges for columns $ i $ and $ j $ are denoted $ \edges^{(i|j)} $ and $ \edges^{(j|i)} $, respectively.
This superscript notation for two-dimensional histograms, $ (i|j) $ and $ (j|i) $, indicates the fact that two-dimensional histograms may have additional bin edges in either or both dimensions compared to one-dimensional histograms due to additional refinement.
One-dimensional histogram bin counts are denoted by the $ \nbins^{(i)} \times \nbins^{(i)} $ diagonal matrix $ \bincounts^{(i)} $ where the diagonal elements contain the bin counts. That is:
\begin{equation}
    \bincounts^{(i)} = \begin{bmatrix}
        \bincount^{(i)}_1 & 0                 & \cdots & 0                              \\
        0                 & \bincount^{(i)}_2 & \cdots & 0                              \\
        \vdots            & \vdots            & \ddots & \vdots                         \\
        0                 & \cdots            & 0      & \bincount^{(i)}_{\nbins^{(i)}}
    \end{bmatrix}
\end{equation}
where $ \bincount^{(i)}_{\binidx} $ denotes the bin count for bin $ \binidx $.
Two-dimensional histogram bin counts are similarly denoted by the $ \nbins^{(i|j)} \times \nbins^{(j|i)} $ matrix $ \bincounts^{(ij)} $.
In this case, off-diagonal values may be non-zero.

\begin{figure}
    \centering
    \footnotesize

\newcommand{\dsfigHistHeight}{8.5em}
\newcommand{\dsfigBlowUpSize}{0.7*\dsfigHistHeight}

\newcommand{\dsfigMinMaxColour}{black}

\pgfplotsset{
    2d histogram plot/.style={
        view={0}{90},
        width=\dsfigHistHeight, height=\dsfigHistHeight,
        ticks=none,
        xticklabels={},
        yticklabels={},
        anchor=south west,
        scale only axis=true,
    },
    /pgfplots/colormap={colour1map}{
        color=(white),
        color=(colour1!20),
        color=(colour1!40),
        color=(colour1!60)
    }
}

\begin{tikzpicture}
    \tikzstyle{every node} = [inner sep=0.6em, text depth=0.3em]

    \begin{axis}[2d histogram plot, name=hist2d]
        \addplot3[surf, shader=flat corner, mesh/cols=10, mesh/ordering=rowwise] file {\dataForPairwiseHistFigureFour};
    \end{axis}
    \node[anchor=south] (labelhist) at (hist2d.north) {\small \textbf{2-d histogram: $ \bincounts^{(ij)} $}};

    \node[anchor=north] (nBins2d1) at (hist2d.south) {$ \nbins^{(j|i)} $};
    \draw[->] (nBins2d1.east) -- (nBins2d1.center -| hist2d.east);
    \draw[->] (nBins2d1.west) -- (nBins2d1.center -| hist2d.west);
    \node[anchor=south, rotate=90] (nBins2d2) at (hist2d.west) {$ \nbins^{(i|j)} $};
    \draw[->] (nBins2d2.east) -- (nBins2d2.east |- hist2d.north);
    \draw[->] (nBins2d2.west) -- (nBins2d2.west |- hist2d.south);

    \node[minimum height=\dsfigBlowUpSize, minimum width=\dsfigBlowUpSize, xshift=5.5em, yshift=-0.8em, postaction={draw, thick, black}, line width=0cm, fill=colour1!20, anchor=west] (blowup) at (hist2d.east) {};
    \node[anchor=south west, inner sep=0, minimum width=0.1111*\dsfigHistHeight, minimum height=0.1111*\dsfigHistHeight, postaction={draw, thick, black}, line width=0cm, xshift=0.6666*\dsfigHistHeight, yshift=0.6666*\dsfigHistHeight] (spotlight) at (hist2d.south west) {};
    \draw[black, thin, densely dashed] (spotlight.north east) -- (blowup.north west);
    \draw[black, thin, densely dashed] (spotlight.south east) -- (blowup.south west);
    \node[anchor=north, align=center] () at (blowup.north |- labelhist.north) {\small \textbf{2-d Bin: $ \bincount^{(ij)}_{\binidx_{i}\binidx_{j}} $}
    };

    \pgfmathsetseed{42}
    \foreach \i in {0,...,45} {
        \pgfmathsetmacro{\x}{(rand + 1) / 2 * 0.81 + 0.08}
        \pgfmathsetmacro{\y}{(rand + 1) / 2 * 0.75 + 0.15}
        \node[circle, fill=black!50, inner sep=0, minimum size=0.35em, xshift=\x*\dsfigBlowUpSize, yshift=\y*\dsfigBlowUpSize] () at (blowup.south west) {};
    }

    \node[anchor=west] () at (blowup.north east) {$ \edge^{(i|j)}_{\binidx_{i}+1} $};
    \node[anchor=west] () at (blowup.south east) {$ \edge^{(i|j)}_{\binidx_{i}} $};
    \node[anchor=north] () at (blowup.south west) {$ \edge^{(j|i)}_{\binidx_{j}} $};
    \node[anchor=north] () at (blowup.south east) {$ \edge^{(j|i)}_{\binidx_{j}+1} $};

    \draw[densely dotted] ([yshift=-0.1*\dsfigBlowUpSize]blowup.north west) -- ([yshift=-0.1*\dsfigBlowUpSize]blowup.north east);
    \draw[densely dotted] ([yshift=0.16*\dsfigBlowUpSize]blowup.south west) -- ([yshift=0.16*\dsfigBlowUpSize]blowup.south east);
    \draw[densely dotted] ([xshift=-0.16*\dsfigBlowUpSize]blowup.north east) -- ([xshift=-0.16*\dsfigBlowUpSize]blowup.south east);
    \draw[densely dotted] ([xshift=0.11*\dsfigBlowUpSize]blowup.north west) -- ([xshift=0.11*\dsfigBlowUpSize]blowup.south west);

    \draw[latex-, densely dotted, \dsfigMinMaxColour] ([yshift=-0.1*\dsfigBlowUpSize]blowup.north east) -- ([yshift=-0.26*\dsfigBlowUpSize, xshift=2.4em]blowup.north east);
    \node[anchor=west, yshift=-0.26*\dsfigBlowUpSize, xshift=2.5em, inner sep=0.2em, \dsfigMinMaxColour] () at (blowup.north east) {$ \binmax{\binidx_{i}}{(i|j)} $};

    \draw[latex-, densely dotted, \dsfigMinMaxColour] ([yshift=0.16*\dsfigBlowUpSize]blowup.south east) -- ([yshift=0.27*\dsfigBlowUpSize, xshift=2.4em]blowup.south east);
    \node[anchor=west, yshift=0.29*\dsfigBlowUpSize, xshift=2.5em, inner sep=0.2em, \dsfigMinMaxColour] () at (blowup.south east) {$ \binmin{\binidx_{i}}{(i|j)} $};

    \node[anchor=south, xshift=-0.14*\dsfigBlowUpSize, inner sep=0.2em, \dsfigMinMaxColour, yshift=0.3em] () at (blowup.north east) {$ \binmax{\binidx_{j}}{(j|i)} $};

    \node[anchor=south, xshift=0.14*\dsfigBlowUpSize, inner sep=0.2em, \dsfigMinMaxColour, yshift=0.3em] () at (blowup.north west) {$ \binmin{\binidx_{j}}{(j|i)} $};

    \draw[] ([xshift=-0.16*\dsfigBlowUpSize, yshift=0.3em]blowup.north east) -- ([xshift=-0.16*\dsfigBlowUpSize, yshift=-0.3em]blowup.north east);

    \draw[] ([xshift=0.11*\dsfigBlowUpSize, yshift=0.3em]blowup.north west) -- ([xshift=0.11*\dsfigBlowUpSize, yshift=-0.3em]blowup.north west);

    \node[anchor=north, xshift=0.475*\dsfigBlowUpSize] () at (blowup.south west) {$ \binmidpoint^{(j|i)}_{ \binidx_{j} } $};
    \node[anchor=west, yshift=0.53*\dsfigBlowUpSize] () at (blowup.south east) {$ \binmidpoint^{(i|j)}_{ \binidx_{i} } $};

    \draw[] ([xshift=0.475*\dsfigBlowUpSize, yshift=0.3em]blowup.south west) -- ([xshift=0.475*\dsfigBlowUpSize, yshift=-0.3em]blowup.south west);

    \draw[] ([xshift=-0.3em, yshift=0.53*\dsfigBlowUpSize]blowup.south east) -- ([xshift=0.3em, yshift=0.53*\dsfigBlowUpSize]blowup.south east);

\end{tikzpicture}
    \caption{
        Notation for two-dimensional histograms with bin counts matrix $ \bincounts^{(ij)} $ and individual bin count $ \bincount^{(ij)}_{\binidx_{i}\binidx_{j}} $.
    }
    \label{fig:histogram}
\end{figure}

\histname{} stores various metadata for each bin, namely:
1) \textit{minimum} and \textit{maximum} actual data values, $ \binmins{\binidx}{(i)} $ and $ \binmaxs{\binidx}{(i)} $;
2) the bin \textit{midpoint}, $ \binmidpoint^{(i)}_{\binidx} $, which is equidistant between the minimum and maximum values;
3) bounds on the weighted centre of the data points within the bin, $ \binmidpoint^{(i)-}_{\binidx} $ and $ \binmidpoint^{(i)+}_{\binidx} $;
and 4) the \textit{unique count}, $ \binuniquecount^{(i)}_ {\binidx} $, which is the number of unique values in the bin.
Note that for two-dimensional histograms, each of these metadata apply to both dimensions, as illustrated in \autoref{fig:histogram}.

In this paper, vectors are indicated by bold lower case letters and matrices by bold upper case letters.
The L1 norm is denoted by $ \lonenorm{\cdot} $, the vector dot product by $ \boldsymbol{x}\cdot\boldsymbol{y} $, Hadamard (element-wise) multiplication by $ \boldsymbol{x}\hadamardproduct\boldsymbol{y} $ and Hadamard division by $ \boldsymbol{x}\hadamarddivision\boldsymbol{y} $.
We also use specific index variables for specific purposes.
That is, $ i $ and $ j $ are used for column indices within a dataset, $ \binidx $ is used for histogram bins and $ \subbinidx $ is used for sub-bins.
The terms \textit{column} and \textit{dimension} are used interchangeably, as well as \textit{tuple} and \textit{data point}.

The following subsections describe \histname{} construction (\autoref{sec:histogram:construction}), bin weighted centre bounds (\autoref{sec:histogram:bin_centres}) and \histname{} storage (\autoref{sec:histogram:storage}).

\subsection{Histogram construction}
\label{sec:histogram:construction}

\begin{figure*}[!t]
    \centering
    \footnotesize
    \begin{tabular}{p{\algowidth}}
	\toprule
	\showAlgoCounter{\label{algo:build}}{\algoBuild} \\ \midrule
	\textbf{Inputs:} dataset~$ \dataset $, sample size~$ \nsamples $, minimum points~$ \minpts $, hypothesis test significance~$ \testsignificance $, initial bin edges~$ \edgesfromgd $ (optional) \\
	\textbf{Outputs:} \histname{} data structure \\
	\vspace{-0.65em}
	\begin{algorithmic}[1]\setstretch{\algolinespacing}
        \State $ \datasetsample \leftarrow $ downsample $ \dataset $ to $ \nsamples $ rows

        \For{$ i = 1, 2, \ldots, $ number of columns in $ \dataset $}
            \State \LineComment{1-d histograms}
            \State $ \edgesinit^{(i)} \leftarrow $ downsample $ \edgesfromgd^{(i)} $ to $ \ceil{\nsamples / \minpts} $ values, else min/max of $ \dataset^{(i)} $
            \State $ \edges^{(i)}, \binmins{}{(i)}, \binmaxs{}{(i)}, \binuniquecounts^{(i)}
            \leftarrow \langle \edgeinit^{(i)}_{0} \rangle, \emptyset, \emptyset, \emptyset $  \algocomment{initialise}

            \For{$ \binidx = 1, 2, \ldots, $ size of $ \edgesinit^{(i)} - 1 $}  \algocomment{loop over initial bins}
                \State $ \datacolumn_{\binidx} \leftarrow $ elements of $ \datasetsample^{(i)} $ between of $ \edgeinit^{(i)}_{\binidx} $ and $ \edgeinit^{(i)}_{\binidx + 1} $

                \State $ \edges^{(i)}_{\text{new}}, \binmins{\text{new}}{(i)}, \binmaxs{\text{new}}{(i)}, \binuniquecounts^{(i)}_{\text{new}}
                \leftarrow \algoRefineOneDim \big( \edgeinit^{(i)}_{\binidx}, \edgeinit^{(i)}_{\binidx + 1}, \datacolumn_{\binidx}, \minpts, \testsignificance \big) $

                \State Append $ \edges^{(i)}_{\text{new}}, \binmins{\text{new}}{(i)}, \binmaxs{\text{new}}{(i)}, \binuniquecounts^{(i)}_{\text{new}} $ to $ \edges^{(i)}, \binmins{}{(i)}, \binmaxs{}{(i)}, \binuniquecounts^{(i)} $
            \EndFor

            \State $ \binmidpoints^{(i)} \leftarrow \big(\binmaxs{}{(i)} + \binmins{}{(i)}\big) / 2 $
            \State $ \bincentreboundslow{}{(i)}, \bincentreboundshigh{}{(i)} \leftarrow $ bin centre bounds~(\autoref{eq:bin_centre_bounds})
            \State $ \bincounts^{(i)} \leftarrow \histfunc \big( \datasetsample^{(i)}, \edges^{(i)} \big) $

            \State \LineComment{2-d histograms}
            \For{$ j = 1, 2, \ldots, i - 1 $}  \algocomment{all columns before $ i $}
                \State $ \edges^{(i|j)}, \edges^{(j|i)} \leftarrow \edges^{(i)}, \edges^{(j)} $  \algocomment{initial 2-d bin edges}
                \State $ \bincounts^{(ij)} \leftarrow \histfunc( \datasetsample^{(ij)}, \edges^{(i|j)}, \edges^{(j|i)} ) $  \algocomment{initial 2-d bin counts}
                \For{each bin $ (\binidx_{i}, \binidx_{j}) $ in $ \bincounts^{(ij)} $ where $ \bincount_{\binidx_{i}\binidx_{j}} > \minpts $}
                    \State $ \datacolumns_{\binidx_{i},\binidx_{j}} \leftarrow $ elements of $ \datasetsample^{(ij)} $ within bin $ (\binidx_{i}, \binidx_{j}) $

                    \State $ \edges^{(i|j)}_{\text{new}}, \edges^{(j|i)}_{\text{new}}
                    \leftarrow \algoRefineTwoDim($\parbox[t]{0.3\columnwidth}{
                        $ \edge^{(i|j)}_{t_i}, \edge^{(i|j)}_{t_i + 1}, \edge^{(j|i)}_{t_j}, \edge^{(j|i)}_{t_j + 1}, \datacolumns_{\binidx_{i},\binidx_{j}}, \minpts, \testsignificance ) $
                    }

                    \State Insert $ \edges^{(i|j)}_{\text{new}} $ into $ \edges^{(i|j)} $ after index $ \binidx_{i} $
                    \State Insert $ \edges^{(j|i)}_{\text{new}} $ into $ \edges^{(j|i)} $ after index $ \binidx_{j} $
                \EndFor
                \State $ \bincounts^{(ij)} \leftarrow \histfunc( \datasetsample^{(ij)}, \edges^{(i|j)}, \edges^{(j|i)} ) $ \algocomment{refined 2-d bin counts}
                \State $ \binmins{}{(i|j)} $, $ \binmins{}{(j|i)} $, $ \binmaxs{}{(i|j)} $, $ \binmaxs{}{(j|i)} \leftarrow $ min/max values in each bin
                \State $ \binmidpoints^{(i|j)}, \binmidpoints^{(j|i)} \leftarrow (\binmaxs{}{(i|j)} + \binmins{}{(i|j)}) / 2, (\binmaxs{}{(j|i)} + \binmins{}{(j|i)}) / 2 $
                \State $ \bincentreboundslow{}{(i|j)}, \bincentreboundshigh{}{(i|j)}, \bincentreboundslow{}{(j|i)}, \bincentreboundshigh{}{(j|i)} \leftarrow $ bin centre bounds~(\autoref{eq:bin_centre_bounds})
                \State $ \binuniquecounts^{(i|j)} $, $ \binuniquecounts^{(j|i)} \leftarrow $ number of unique values in each bin
            \EndFor
        \EndFor
	\end{algorithmic}
	\\[-1em] \bottomrule
\end{tabular}

\end{figure*}

\histname{} construction is outlined in \autoref{algo:build}, which consists of three sections:
1) extract a sample $ \datasetsample $ of size $ \nsamples $ from dataset $ \dataset $ (line~1),
2) iterate over each column to generate one-dimensional histograms (lines~3--11), and
3) iterate over each pair of columns to generate two-dimensional histograms (lines~13--26).

One-dimensional histograms are generated by first selecting initial bin edges, $ \edgesinit $, using either the bases from GreedyGD (downsampled to at most $ \ceil*{\nsamples / M} $) or just the min and max values of the relevant column (line~4).
Each initial bin is then refined using \algoRefineOneDim{} (lines~6--9), which determines whether a bin should be split or not based on a hypothesis test.
The one-dimensional histograms are finalised by computing the midpoints, weighted centre bounds and bin counts (using a standard histogram function, denoted $ \histfunc $) in lines~10--12.


\begin{figure*}[!t]
    \centering
    \footnotesize
    \begin{tabular}{p{\algowidth}}
    \toprule
    \showAlgoCounter{\label{algo:refine1d}}{\algoRefineOneDim} \\ \midrule
    \textbf{Inputs:}
        bin lower edge~$ \edge_{\low} $,
        bin upper edge~$ \edge_{\high} $,
        vector of data values~$ \datacolumn $ within the bin,
        minimum points~$ \minpts $,
        significance~$ \testsignificance $ \\
    \textbf{Outputs:}
        upper bin edges~$ \edgesrefined{}{} $,
        bin minimum values~$ \binmins{}{} $,
        bin maximum values~$ \binmaxs{}{} $,
        bin unique counts~$ \binuniquecounts $ \\
    \vspace{-0.65em}
    \begin{algorithmic}[1]\setstretch{\algolinespacing}
        \State $ \datasetuniques \leftarrow $ unique values in $ \datacolumn $
        \State $ \nunique \leftarrow $ number of elements in $ \datasetunique $
        \If{$ \datacolumn $ is empty}  
            \State \Return
            $ \langle \edge_{\high} \rangle $,
            $ \langle \edge_{\low} \rangle $,
            $ \langle \edge_{\high} \rangle $,
            $ \langle 0 \rangle $

        \ElsIf{$ \nunique = 1 $}  
            \State \Return
            $ \langle \edge_{\high} \rangle $,
            $ \langle \datasetunique_{0} \rangle $,
            $ \langle \datasetunique_{0} \rangle $,
            $ \langle 1 \rangle $

        \ElsIf{fewer than $ \minpts $ tuples in $ \datacolumn $ or $ \algoIsUniform\!\left( \datacolumn, \edge_{\low}, \edge_{\high}, \nunique, \testsignificance \right) $}  
            \State \Return
            $ \langle \edge_{\high} \rangle $,
            $ \langle \min(\datasetuniques) \rangle $,
            $ \langle \max(\datasetuniques) \rangle $,
            $ \langle \nunique \rangle $

        \Else  
            \State $ \splitpoint \leftarrow $ select split point

            \State $ \datacolumn_{\low},\, \datacolumn_{\high} \leftarrow $ split $ \datacolumn $ at $ \splitpoint $

            \State $ \edgesrefined{\low}{},\, \binmins{\low}{},\, \binmaxs{\low}{},\, \binuniquecounts_{\low}^{\,} \leftarrow \algoRefineOneDim( \datacolumn_{\low},\, \edge_{\low},\, \splitpoint, \minpts,\, \testsignificance ) $

            \State $ \edgesrefined{\high}{},\, \binmins{\high}{},\, \binmaxs{\high}{},\, \binuniquecounts_{\high}^{\,} \leftarrow \algoRefineOneDim( \datacolumn_{\high},\, \splitpoint,\, \edge_{\high}, \minpts,\, \testsignificance ) $

            \State \Return
            $ \langle \edgesrefined{\low}{}, \edgesrefined{\high}{} \rangle $,
            $ \langle \binmins{\low}{}, \binmins{\high}{} \rangle $,
            $ \langle \binmaxs{\low}{}, \binmaxs{\high}{} \rangle $,
            $ \langle \boldsymbol{u}_{\low}, \boldsymbol{u}_{\high} \rangle $
        \EndIf
    \end{algorithmic}
    \\[-1em] \bottomrule
\end{tabular}

\end{figure*}

\algoRefineOneDim{} is described in detail in \autoref{algo:refine1d}.
This is a recursive algorithm that checks whether a given bin needs to be split, i.e., if the distribution of data points within the bin is not sufficiently uniform.
If so, it splits the bin and calls itself on the two newly created splits, denoted by $ \low $ and $ \high $.
We tested both equal-width (split at bin midpoint) and equal-depth (split at median) approaches and found equal-width to perform slightly better.
Bins will not be split if they are empty (line~3), have only one unique value (line~5) or have too few data points (line~7).
\algoRefineOneDim{} returns the upper edges from the original bin and all new splits, as well as the bin minimum(s), maximum(s) and unique count(s).

\begin{figure}[!t]
    \centering
    \small

\newcommand{\figbrwidth}{1/9*\columnwidth}
\newcommand{\figbrheight}{1/9*\columnwidth}
\newcommand{\figbrspacing}{3/9/5*\columnwidth}
\newcommand{\figbrpointsize}{0.15em}
\newcommand{\figbrlightwidth}{0.3em}

\begin{tikzpicture}[
        grid node/.style={
            minimum height=\figbrwidth,
            minimum width=\figbrheight,
            anchor=north west,
            draw,
            outer sep=0pt,
        },
        data points/.style={
            circle,
            fill=black!40,
            inner sep=0,
            minimum size=\figbrpointsize,
            text height=0,
        },
        data grid/.pic={
            \node[grid node] (g1) at (0,0) {};
            \node[grid node] (g2) at (g1.south west) {};
            \node[grid node] (g3) at (g2.south west) {};
            \foreach \i in {0,...,100} {
                \pgfmathsetmacro{\x}{(rand + 1) / 2 * 0.96 + 0.02}
                \pgfmathsetmacro{\y}{(rand + 1) / 2 * 0.96 + 0.02}
                \node[data points, xshift=\x*\figbrwidth, yshift=\y*\figbrheight] () at (g1.south west) {};
            };
            \pgfmathsetseed{400}
            \foreach \i in {0,...,100} {
                \pgfmathsetmacro{\x}{int(((rand + 1) / 2) ^ (0.4) * 96 + 2)}
                \pgfmathsetmacro{\y}{int((rand + 1) / 2 * 96 + 2)}
                \node[data points, xshift=\x/100*\figbrwidth, yshift=\y/100*\figbrheight] () at (g2.south west) {};
            };
            \foreach \i in {0,...,50} {
                \pgfmathsetmacro{\x}{int((rand + 1) / 2 * 96 + 2)}
                \pgfmathsetmacro{\y}{int(((rand + 1) / 2) ^ (0.4) * 96 + 2)}
                \node[data points, xshift=\x/100*\figbrwidth, yshift=\y/100*\figbrheight] () at (g2.south west) {};
            };
            \foreach \i in {0,...,15} {
                \pgfmathsetmacro{\x}{(rand + 1) / 2 * 0.96 + 0.02}
                \pgfmathsetmacro{\y}{(rand + 1) / 2 * 0.96 + 0.02}
                \node[data points, xshift=\x*\figbrwidth, yshift=\y*\figbrheight] () at (g3.south west) {};
            };
        },
        bin highlights/.style={
            minimum height=\figbrheight-\figbrlightwidth,
            minimum width=\figbrwidth-\figbrlightwidth,
            draw,
            anchor=north west,
            line width=\figbrlightwidth,
            inner sep=0pt,
            opacity=0.4,
        },
        subfigure labels/.style={
            text depth=0.3em,
            text height=1.0em,
            anchor=south,
        },
        new split/.style={colour2, ultra thick, densely dashed},
        old split/.style={black},
    ]

    \path (0,0) pic {data grid};
    \path (1*\figbrwidth + 1*\figbrspacing,0) pic {data grid};
    \path (2*\figbrwidth + 2*\figbrspacing,0) pic {data grid};
    \path (3*\figbrwidth + 3*\figbrspacing,0) pic {data grid};
    \path (4*\figbrwidth + 4*\figbrspacing,0) pic {data grid};
    \path (5*\figbrwidth + 5*\figbrspacing,0) pic {data grid};

    \node[subfigure labels] () at (0.5*\figbrwidth, 0.2em) {(a)};
    \node[subfigure labels] () at (1.5*\figbrwidth+1*\figbrspacing, 0.2em) {(b)};
    \node[subfigure labels] () at (2.5*\figbrwidth+2*\figbrspacing, 0.2em) {(c)};
    \node[subfigure labels] () at (3.5*\figbrwidth+3*\figbrspacing, 0.2em) {(d)};
    \node[subfigure labels] () at (4.5*\figbrwidth+4*\figbrspacing, 0.2em) {(e)};
    \node[subfigure labels] () at (5.5*\figbrwidth+5*\figbrspacing, 0.2em) {(f)};

    \draw[->, ultra thick] ([xshift=3pt]1*\figbrwidth, -0.5*\figbrheight) -- +([xshift=-6pt]1*\figbrspacing,0);
    \draw[->, ultra thick] ([xshift=3pt]2*\figbrwidth+1*\figbrspacing, -1.5*\figbrheight) -- +([xshift=-6pt]1*\figbrspacing,0);
    \draw[->, ultra thick] ([xshift=3pt]3*\figbrwidth+2*\figbrspacing, -1.5*\figbrheight) -- +([xshift=-6pt]1*\figbrspacing,0);
    \draw[->, ultra thick] ([xshift=3pt]4*\figbrwidth+3*\figbrspacing, -1.5*\figbrheight) -- +([xshift=-6pt]1*\figbrspacing,0);
    \draw[->, ultra thick] ([xshift=3pt]5*\figbrwidth+4*\figbrspacing, -2.5*\figbrheight) -- +([xshift=-6pt]1*\figbrspacing,0);

    \node[bin highlights, colour3] () at (1*\figbrwidth+1*\figbrspacing,0) {};

    \node[bin highlights, colour2] () at (2*\figbrwidth+2*\figbrspacing,-1*\figbrheight) {};

    \node[bin highlights, colour2, minimum width=0.5*\figbrwidth-\figbrlightwidth] () at (3*\figbrwidth+3*\figbrspacing,-1*\figbrheight) {};
    \node[bin highlights, colour3, minimum width=0.5*\figbrwidth-\figbrlightwidth] () at (3.5*\figbrwidth+3*\figbrspacing,-1*\figbrheight) {};

    \node[bin highlights, colour3, minimum width=0.5*\figbrwidth-\figbrlightwidth, minimum height=0.5*\figbrheight-\figbrlightwidth] () at (4*\figbrwidth+4*\figbrspacing,-1*\figbrheight) {};
    \node[bin highlights, gray, minimum width=0.5*\figbrwidth-\figbrlightwidth, minimum height=0.5*\figbrheight-\figbrlightwidth] () at (4*\figbrwidth+4*\figbrspacing,-1.5*\figbrheight) {};

    \node[bin highlights, gray, minimum width=0.5*\figbrwidth-\figbrlightwidth] () at (5*\figbrwidth+5*\figbrspacing,-2*\figbrheight) {};
    \node[bin highlights, gray, minimum width=0.5*\figbrwidth-\figbrlightwidth] () at (5.5*\figbrwidth+5*\figbrspacing,-2*\figbrheight) {};

    \draw[new split] (2.5*\figbrwidth+2*\figbrspacing, 0) -- +(0, -3*\figbrheight);
    \draw[old split] (3.5*\figbrwidth+3*\figbrspacing, 0) -- +(0, -3*\figbrheight);
    \draw[old split] (4.5*\figbrwidth+4*\figbrspacing, 0) -- +(0, -3*\figbrheight);
    \draw[old split] (5.5*\figbrwidth+5*\figbrspacing, 0) -- +(0, -3*\figbrheight);

    \draw[new split] (3*\figbrwidth+3*\figbrspacing, -1.5*\figbrheight) -- +(\figbrwidth,0);
    \draw[old split] (4*\figbrwidth+4*\figbrspacing, -1.5*\figbrheight) -- +(\figbrwidth,0);
    \draw[old split] (5*\figbrwidth+5*\figbrspacing, -1.5*\figbrheight) -- +(\figbrwidth,0);

\end{tikzpicture}
    \vspace{-0.5em}
    \caption{
        Illustration of two-dimensional bin refinement:
        (a)~original data;
        (b)~first bin is uniformly distributed (green);
        (c)~second bin is non-uniformly distributed (red) in both dimensions, but less uniform in the vertical dimension, so a new split is added to all bins in the same column;
        (d)~one of the resulting sub-bins is non-uniform in the vertical dimension, so a new split is added;
        (e)~the resulting sub-bins are uniform (green) or contain fewer than $ \minpts $ points (gray), no further splits;
        (f)~both sub-bins from the third original bin contain fewer than $ \minpts $ points, no further splits.
    }
    \label{fig:2d_bin_refinement}
\end{figure}
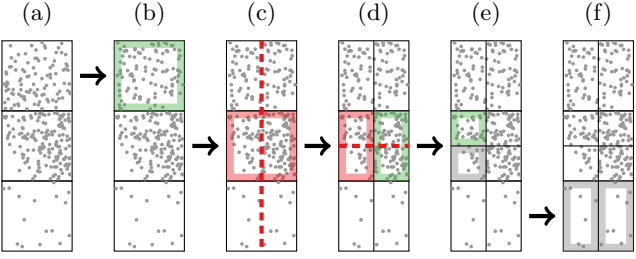

Two-dimensional histograms are constructed using an inner loop in \algoBuild{} (line~14) that iterates over all columns previously iterated over by the outer loop (line~2), thereby covering all pairs of columns.
For each column pair, an initial histogram is created using bin edges from the corresponding one-dimensional histograms (lines~15--16).
This histogram is then refined by applying \algoRefineTwoDim{} to each bin with at least $ \minpts $ tuples (lines~17--21).
\algoRefineTwoDim{} is the two-dimensional analogue of \algoRefineOneDim{} and performs a hypothesis test for uniformity on each column separately.
In the case where both columns are non-uniform, the split is applied to the least uniform column.
Note that any bin splits created by \algoRefineTwoDim{} apply only to the current column pair.
That is, if a split is applied to bin $ \binidx_{i} $ in the $ (ij) $th histogram, it does not affect any other histograms involving column $ i $.
However, the split \textit{does} apply to all bins with the same bin index $ \binidx_{i} $ in the $ (ij) $th histogram.
\autoref{fig:2d_bin_refinement} provides an illustration of the two dimensions bin refinement process.
The two-dimensional histograms are completed by calculating the bin counts, minimums, maximums, midpoints, weighted centre bounds and unique counts (lines~22--26).

Hypothesis testing in both \algoRefineOneDim{} and \algoRefineTwoDim{} is performed using the function \algoIsUniform{}.
Specifically, a chi-squared test is performed against the null hypothesis that data points in the given bin are uniformly distributed between the bin edges.
That is, the bin is divided into a number of sub-bins and the number of points in each sub-bin is compared to the expected number under the null hypothesis.
The appropriate number of sub-bins, $ \nsubbins $, is determined using the Terrell-Scott inequality~\citep{Scott_2009} as follows:
\begin{align}
    \nsubbins = \ceil*{\left( 2 \binuniquecount \right)^{1/3}}.
\end{align}
The test statistic is then:
\begin{align}
    \chisq{} = \sum_{\subbinidx=0}^{\nsubbins-1} \frac{\big(\subbincount_{\subbinidx} - \subbincountexpected \big)^2}{\subbincountexpected},
\end{align}
where $ \subbincountexpected = \bincount / \nsubbins $ is the expected sub-bin count under the null hypothesis and $ \subbincount_{\subbinidx} $ is the actual count for the $ \subbinidx $th sub-bin.
The critical value, $ \chisq{\testsignificance} $, is then defined such that $ \prob\left(\chisq{} > \chisq{\testsignificance} \right) = \testsignificance $ and the null hypothesis is rejected (and the bin split) if $ \chisq{} > \chisq{\testsignificance} $.

Notably, \histname{} construction is highly parallelisable, since each histogram and bin refinement can be computed independently, provided one-dimensional histograms are constructed first.


\begin{figure*}[!t]
    \centering
    \tiny
    \newcommand{\figstoragebytewidth}{1.3em}

\tikzset{
    every node/.style={
        text depth=0.3em,
        text height=1.0em,
        align=left,
        inner sep=0.5em,
    },
    storage node/.style={
        minimum width=\figstoragebytewidth,
        align=center,
        fill=colour1!15,
        anchor=west,
        line width=0cm,
        postaction={draw, black},
        inner sep=0.35em
    },
    params node/.style={
        storage node,
        fill=black!15,
    },
    dot node/.style={
        text depth=0,
        text height=0,
        inner sep=0.13em,
        circle,
    },
    connector lines/.style={
        densely dashed,
        black!60,
    }
}

\begin{tikzpicture}
    \matrix[] (main) {
        \node[params node, minimum height=2em] (centralparams) {\textbf{Params}};
        \node[storage node, minimum height=2em, minimum width=10*\figstoragebytewidth] (central1d) at (centralparams.east) {\textbf{1-d hists}};
        \node[storage node, minimum height=2em, minimum width=5*\figstoragebytewidth] (central2d) at (central1d.east) {\textbf{2-d hists}};
        \node[storage node, minimum height=2em, minimum width=22*\figstoragebytewidth, fill=colour2!15] (centralcounts) at (central2d.east) {\textbf{Bin counts}};
        \node[anchor=east, xshift=-1em, inner sep=0] () at (centralparams.west) {\textbf{Overall Storage Configuration}};
        \\
    };

    \begin{scope}[local bounding box=group 1]
        \matrix[inner sep=0, anchor=south, yshift=1.8*\figstoragebytewidth, xshift=-4.5*\figstoragebytewidth] (params) at (centralparams.north) {
            \node[params node, minimum width=3*\figstoragebytewidth] (nrows) {$ \nrows $};
            \node[params node, minimum width=3*\figstoragebytewidth] (nsamples) at (nrows.east) {$ \nsamples $};
            \node[params node, minimum width=3*\figstoragebytewidth] (minpts) at (nsamples.east) {$ \minpts $};
            \node[params node, minimum width=1.5*\figstoragebytewidth] (significance) at (minpts.east) {$ \testsignificance $};
            \node[params node, minimum width=0.75*\figstoragebytewidth] (dims) at (significance.east) {$ \ndims $};
            \node[params node, minimum width=5*\figstoragebytewidth] (bytedepths) at (dims.east) {$ \bytesrequired^{(1)}, \ldots, \bytesrequired^{(d)} $};
            \\
        };
        \node[anchor=south] () at (params.north) {Parameters: \hspace{0.2em} $ \storage_{\text{params}} = 29 + \ndims $ bytes};
    \end{scope}

    \draw[connector lines] (centralparams.north west) -- (nrows.south west) node [pos=0, dot node] {} node [pos=1, dot node] {};
    \draw[connector lines] (centralparams.north east) -- (bytedepths.south east) node [pos=0, dot node] {} node [pos=1, dot node] {};

    \begin{scope}[local bounding box=group 2]
        \matrix[inner sep=0, anchor=south, yshift=1.8*\figstoragebytewidth, xshift=8*\figstoragebytewidth] (hists2d) at (central2d.north) {
            \node[params node, minimum width=4.5*\figstoragebytewidth] (hist2dnew) {$ \nbins^{(i|j)} - \nbins^{(i)} $};
            \node[storage node, minimum width=6.2*\figstoragebytewidth] (hist2dedges) at (hist2dnew.east) {$ \ldots, \edge^{(i|j)}_{\binidx_i}, \ldots $};
            \node[storage node, minimum width=6.2*\figstoragebytewidth] (hist2dmins) at (hist2dedges.east) {$ \ldots, \binmin{\binidx_i}{(i|j)}, \ldots $};
            \node[storage node, minimum width=6.2*\figstoragebytewidth] (hist2dmaxs) at (hist2dmins.east) {$ \ldots, \binmax{\binidx_i}{(i|j)}, \ldots $};
            \node[storage node, minimum width=6.2*\figstoragebytewidth] (hist2dunique) at (hist2dmaxs.east) {$ \ldots, \binuniquecount^{(i|j)}_{\binidx_i}, \ldots $};
            \\
        };
        \node[anchor=south] () at (hists2d.north) {2-d Histograms: \hspace{0.2em} $ \storage_{\text{2-d hists}} = 2\ndims(\ndims-1) + \sum_{i=1}^{\ndims} (3\bytesrequired^{(i)} + 4) \big( \sum_{j=1}^{\ndims} \nbins^{(i|j)} - \ndims \nbins^{(i)} \big) $ bytes};  
    \end{scope}

    \draw[connector lines] (central2d.north west)+(2.0*\figstoragebytewidth,0) -- (hists2d.south west) node [pos=0, dot node] {} node [pos=1, dot node] {};
    \draw[connector lines] (central2d.north west)+(2.8*\figstoragebytewidth,0) -- (hists2d.south east) node [pos=0, dot node] {} node [pos=1, dot node] {};

    \begin{scope}[local bounding box=group 3]
        \matrix[inner sep=0, anchor=north west, yshift=-1.8*\figstoragebytewidth] (hists1d) at (central1d.south -| group 1.west) {
            \node[params node, minimum width=1.5*\figstoragebytewidth] (hist2dnew) {$ \nbins^{(i)} $};
            \node[storage node, minimum width=5.3*\figstoragebytewidth] (hist2dedges) at (hist2dnew.east) {$ \ldots, \edge^{(i)}_{\binidx}, \ldots $};
            \node[storage node, minimum width=5.3*\figstoragebytewidth] (hist2dmins) at (hist2dedges.east) {$ \ldots, \binmin{\binidx}{(i)}, \ldots $};
            \node[storage node, minimum width=5.3*\figstoragebytewidth] (hist2dmaxs) at (hist2dmins.east) {$ \ldots, \binmax{\binidx}{(i)}, \ldots $};
            \node[storage node, minimum width=5.3*\figstoragebytewidth] (hist2dunique) at (hist2dmaxs.east) {$ \ldots, \binuniquecount^{(i)}_{\binidx}, \ldots $};
            \\
        };
        \node[anchor=north] () at (hists1d.south) {1-d Histograms: \hspace{0.2em} $ \storage_{\text{1-d hists}} = 2\ndims + \sum_{i=1}^{\ndims} \nbins^{(i)} (3\bytesrequired^{(i)} + 4) $ bytes};  
    \end{scope}
    \draw[connector lines] (central1d.south west)+(2.0*\figstoragebytewidth,0) -- (hists1d.north west) node [pos=0, dot node] {} node [pos=1, dot node] {};
    \draw[connector lines] (central1d.south west)+(2.8*\figstoragebytewidth,0) -- (hists1d.north east) node [pos=0, dot node] {} node [pos=1, dot node] {};

    \begin{scope}[local bounding box=group 4]
        \matrix[inner sep=0, anchor=north, yshift=-1.8*\figstoragebytewidth, xshift=-3*\figstoragebytewidth] (countsdense) at (centralcounts.south) {
            \node[params node, minimum width=0.75*\figstoragebytewidth] (lencountsdense) {$ \lencount^{(ij)} $};
            \node[params node, minimum width=0.75*\figstoragebytewidth] (countsindicatordense) at (lencountsdense.east) {$ \countstorageindicator^{(ij)} $};
            \node[storage node, minimum width=5.5*\figstoragebytewidth, fill=colour2!15] (countvaluesdense) at (countsindicatordense.east) {$ \ldots, \bincount_{\scalebox{0.65}{$ \binidx_i \binidx_j $}}^{(ij)}, \ldots $};

            \node[anchor=west, xshift=1em, inner sep=0, align=left, yshift=-0.7em] () at (countvaluesdense.east) {
                Dense: \\[-0.6em]
                $ \storage_{\text{counts}} = 2\ndims^2 + \sum\limits_{i=1}^{\ndims}\sum\limits_{j=1}^{\ndims} \Big\lceil{ \frac{\nbins^{\scalebox{0.45}{$ (i|j) $}} \nbins^{\scalebox{0.45}{$ (j|i) $}} \lencount^{\scalebox{0.45}{$ (ij) $}}}{8} \Big\rceil} $};
            \\
        };

        \matrix[inner sep=0, anchor=north west, yshift=-1.4em] (countssparse) at (countsdense.south west) {
            \node[params node, minimum width=0.75*\figstoragebytewidth] (lencountssparse) {$ \lencount^{(ij)} $};
            \node[params node, minimum width=0.75*\figstoragebytewidth] (countsindicatorsparse) at (lencountssparse.east) {$ \countstorageindicator^{(ij)} $};
            \node[params node, minimum width=0.75*\figstoragebytewidth] (countsnonzero) at (countsindicatorsparse.east) {$ \nbinsnonzero^{(ij)} $};
            \node[storage node, minimum width=5.5*\figstoragebytewidth, fill=colour2!15] (countvaluessparse) at (countsnonzero.east) {$ \ldots, \big\{ \golomb(i \nbins^{(i)} + j), \bincount_{\scalebox{0.65}{$ \binidx_i \binidx_j $}}^{(ij)} \big\}, \ldots $};

            \node[anchor=west, xshift=1em, inner sep=0, align=left, yshift=-0.7em] () at (countvaluessparse.east) {
                Sparse: \\[-0.7em]
                $ \storage_{\text{counts}} = 3\ndims^2 + \sum\limits_{i=1}^{\ndims}\sum\limits_{j=1}^{\ndims} \Big\lceil{ \frac{\storage_{\golomb} + \nbinsnonzero \lencount}{8} \Big\rceil} $};
            \\
        };

        \node[anchor=west, inner sep=0em, outer sep=0em, yshift=-0.7em] () at (lencountsdense.south west) {or};
    \end{scope}
    \draw[connector lines] (centralcounts.south west)+(5.0*\figstoragebytewidth,0) -- (countsdense.north west) node [pos=0, dot node] {} node [pos=1, dot node] {};
    \draw[connector lines] (centralcounts.south west)+(5.8*\figstoragebytewidth,0) -- (countvaluesdense.north east) node [pos=0, dot node] {} node [pos=1, dot node] {};

    \scoped[on background layer] {
        \node[anchor=center, fill=black!5, rounded corners, minimum height=4.2em, minimum width=53*\figstoragebytewidth] () at (main) {};
    }

\end{tikzpicture}
    \caption{
        \histname{} storage configuration.
    }
    \label{fig:storage}
\end{figure*}


\subsection{Bin weighted centre bounds}
\label{sec:histogram:bin_centres}

To improve query bounds, \histname{} stores weighted centres bounds for each histogram bin.
The value of the weighted centre bounds depends on whether bins passed the hypothesis test or not.
For bins that did not pass, the only available information on their internal data distribution is that they contain $ \bincount $ data points, $ \binuniquecount $ unique values and extrema $ \binmin{}{} $ and $ \binmax{}{} $.
In this case, weighted centre bounds occur when $ \bincount - \binuniquecount + 1 $ points are at an extrema and one point is at each of the other unique values, which are assumed to be as close to the extrema as possible, i.e., with minimum spacing for the given data type, denoted $ \datatypequanta $.
Conversely, the internal distribution of bins that pass the hypothesis test is known to  be approximately uniform with respect to a given number of sub-bins.
This fact can be used to derive tighter bounds, as shown in \autoref{theorem:centre_bounds}.
Passing and non-passing bins can be distinguished by their bin count, which is at least $ \minpts $ for passing bins and less than $ \minpts $ for non-passing bins.

\begin{theorem}
    \label{theorem:centre_bounds}
    Consider a bin with count $ \bincount $, minimum value $ \binmin{}{} $, maximum value $ \binmax{}{} $ and $ \binuniquecount $ unique values.
    Assume this bin satisfies the hypothesis test in \textup{\algoIsUniform{}} with $ \nsubbins = \big\lceil(2\binuniquecount)^{1/3} \big\rceil $ sub-bins and critical value $ \chisq{\testsignificance} $.
    Let $ \subbinwidth = (\binmax{}{} - \binmin{}{}) / \nsubbins $ be the sub-bin width.
    The bounds for the weighted centre of the points within the bin are then
    \begin{equation}
        c^{\pm} = \binmin{}{} + \frac{(s\pm1) \subbinwidth}{2} \pm \frac{\subbinwidth}{6} \sqrt{\frac{3 \chisq{\testsignificance} (s^2 - 1)}{h}}.
    \end{equation}
\end{theorem}
\begin{proof}
    The lower bound occurs when all points are at the lower edge of their respective sub-bins and can be expressed as follows:
    \begin{equation} \label{eq:centre_weighted_min_initial}
        \bincentreboundlow{}{} = \frac{1}{\bincount} \sum_{\subbinidx=0}^{\nsubbins - 1} \subbincount_{\subbinidx} (\binmin{}{} + \subbinidx \subbinwidth).
    \end{equation}
    where $ \subbincount_{\subbinidx} $ is the count for the $ \subbinidx $th sub-bin.
    Sub-bin counts can be expressed in terms of the expected sub-bin count plus an epsilon term, i.e., $ \subbincount_{\subbinidx} = \bincount/\nsubbins + \subbincountdelta_{\subbinidx} $.
    Substituting this into \autoref{eq:centre_weighted_min_initial} gives:
    \begin{equation} \label{eq:centre_weighted_min_with_epsilon}
        \bincentreboundlow{}{} = \binmin{}{} + \frac{\subbinwidth (\nsubbins-1)}{2} + \frac{\subbinwidth}{\bincount} \sum_{\subbinidx=0}^{\nsubbins-1} \subbinidx \subbincountdelta_{\subbinidx}.
    \end{equation}
    This can be optimised using Lagrange Multipliers with constraints:
    \begin{equation} \label{eq:theorem1:constraints}
        \sum_{\subbinidx=0}^{\nsubbins-1} \subbincountdelta_{\subbinidx} = 0
        \text{ and }
        \chisq{\testsignificance} = \sum_{\subbinidx=0}^{\nsubbins-1}\frac{(\subbincount_{\subbinidx} - \bincount/\nsubbins)^2}{\bincount/\nsubbins} = \frac{\nsubbins}{\bincount} \sum_{\subbinidx=0}^{\nsubbins-1} \subbincountdelta_{\subbinidx}^2.
    \end{equation}
    The Lagrangian is then:
    \begin{align}
        \mathcal{L}(\subbincountdelta_{\subbinidx}, \lambda_1, \lambda_2) = ~
        \binmin{}{} & + \frac{\subbinwidth (\nsubbins-1)}{2} + \frac{\subbinwidth}{\bincount} \sum_{\subbinidx=0}^{\nsubbins-1} \subbinidx \subbincountdelta_{\subbinidx} \nonumber\\[-0.2em]
        & + \lambda_1 \sum_{\subbinidx=0}^{\nsubbins-1} \subbincountdelta_{\subbinidx} + \lambda_2 \left(\chisq{\testsignificance} - \frac{\nsubbins}{\bincount} \sum_{\subbinidx=0}^{\nsubbins-1} \subbincountdelta_{\subbinidx}^2\right).
    \end{align}
    Setting the partial derivative $ \partial \mathcal{L} / \partial \subbincountdelta_{\subbinidx} $ equal to zero gives $ \subbincountdelta_{\subbinidx} = (\subbinidx \subbinwidth + \bincount \lambda_1) / 2s\lambda_2 $.
    Substituting this into the constraints in \autoref{eq:theorem1:constraints} and solving for $ \lambda_1 $ and $ \lambda_2 $ gives the following expression for $ \subbincountdelta_{\subbinidx} $:
    \begin{equation}
        \subbincountdelta_{\subbinidx} = \pm \frac{2}{\nsubbins} \left(\subbinidx - \frac{\nsubbins(\nsubbins-1)}{2}\right) \sqrt{\frac{3 \chisq{\testsignificance} \bincount}{\nsubbins^2 - 1}}.
    \end{equation}
    Taking the negative solution and substituting this into \autoref{eq:centre_weighted_min_with_epsilon} gives the desired result.
    A similar approach can be used for~$ \binmax{}{} $.
\end{proof}

Thus, bin weighted centre bounds are as follows:
\begin{equation}  \label{eq:bin_centre_bounds}
    \binmidpoint^{\pm} =
    \begin{cases}
        \binminmax^{\pm} \mp \frac{(\binuniquecount - 1) \binuniquecount \datatypequanta}{2 \bincount}, & \bincount < \minpts \\
        \binmin{}{} + \frac{(\nsubbins \pm 1) \subbinwidth}{2} \pm \frac{\subbinwidth}{6} \sqrt{\frac{3 \chisq{\testsignificance} (\nsubbins^2 - 1)}{\bincount}}, & \text{otherwise.}
    \end{cases}
\end{equation}

\subsection{Storage}
\label{sec:histogram:storage}

To minimise \histname{} storage requirements, we observe that bin midpoints and weighted centre bounds can easily be re-derived from other parameters and thus need not be stored.
Additionally, bin counts, which require the most storage, are stored sparsely when this is more effective.
For sparse encoding, we store the delta between non-zero indices and encode using Golomb coding, which is optimal for geometrically distributed data.
\autoref{fig:storage} provides an overview of the storage configuration.
In total, the storage requirements are:
\begin{align}
    \storage &= \storage_{\text{params}} + \storage_{\text{1-d hists}} + \storage_{\text{2-d hists}} + \storage_{\text{counts}} \\
    &\le 29 + \ndims + 4\ndims^2 \nonumber\\[-0.4em]
    & \hspace{2em} + \sum_{i=1}^{\ndims} \big(3\bytesrequired^{(i)} + 4 \big) \Bigg( \sum_{j=1}^{d} \nbins^{(i|j)} - (\ndims - 1) \nbins^{(i)} \Bigg) \nonumber\\[-0.2em]
    & \hspace{2em} + \sum_{i=1}^{\ndims}\sum_{j=1}^{\ndims} \big\lceil{ \nbins^{(i|j)} \nbins^{(j|i)} \lencount^{(ij)} \ / \ 8 \big\rceil} \text{ bytes,}
\end{align}
where $ \lencount $ is the number of bits per bin count, i.e.,
\begin{equation}
    \lencount^{(ij)} = \Big\lceil \log_2 \Big(1 + \max_ {\binidx_{i},\binidx_{j}}{\left(\bincount^{(ij)}_{\binidx_{i}\binidx_{j}}\right)}\Big) \Big\rceil,
\end{equation}
and $ \bytesrequired^{(i)} $ is the number of bytes per value in the $ i $th dimension.
In \autoref{fig:storage}, $ \countstorageindicator^{(ij)} $ is a binary variable that indicates whether the $ (ij) $th histogram is stored densely or sparsely, $ \nbinsnonzero^{(ij)} $ is the number of non-zero values in $ \bincounts^{(ij)} $ and $ \golomb(\cdot) $ is Golomb coding.

    \section{Query Execution}
\label{sec:query_execution}

\begin{figure*}[!t]
    \centering
    \footnotesize
    \newcommand{\figqueryvspacing}{0.9em}
\newcommand{\figqueryoperatorwidth}{3.1em}
\newcommand{\figqueryhistcolour}{colour1!30}
\newcommand{\figqueryhighlightcolour}{colour1!10}

\tikzset{
    every node/.style={
        text depth=0.3em,
        text height=1.0em,
        inner sep=0.4em,
        align=left,
        anchor=west
    },
    predicate node/.style={
        draw, rounded corners, fill=black!5
    },
    preprocessing node/.style={
        anchor=east,
        align=right,
        text depth=0.0em,
        text height=0.5em,
        yshift=-\figqueryvspacing-1.0em,
        font=\tiny
    },
}

\pgfplotsset{
    1d histogram plot/.style={
        ybar,
        width=6em, height=2.5em,
        axis background/.style={fill=white},
        ymin=0,
        grid=none,
        ticks=none,
        scale only axis=true,
        axis on top=true,
    },
}

\begin{tikzpicture}
    \matrix[matrix anchor=center, inner sep=0em] (query) at (0,0) {
        \node[predicate node, label=above:{$ \predicate_1 $}] (p1) {\textit{dist} $ > $ 150};
        \node[minimum width=\figqueryoperatorwidth] (and1) at (p1.east) {AND};
        \node[predicate node, label=above:{$ \predicate_2 $}] (p2) at (and1.east) {\textit{dist} $ < $ 300 };
        \node[minimum width=\figqueryoperatorwidth] (or1) at (p2.east) {OR};
        \node[predicate node, label=above:{$ \predicate_3 $}] (p3) at (or1.east) {\textit{dist} $ < $ 450 };
        \node[minimum width=\figqueryoperatorwidth] (and2) at (p3.east) {AND};
        \node[predicate node, label=above:{$ \predicate_4 $}] (p4) at (and2.east) {\textit{air time} $ > $ 90.5};
        \\
    };

    \node[predicate node, anchor=north, yshift=-2*\figqueryvspacing-1.4em] (p1t) at (p1.south) {$ \colvar_2 > 81 $};
    \node[predicate node, anchor=center] (p2t) at (p1t -| p2) {$ \colvar_2 > 231 $};
    \node[predicate node, anchor=center] (p3t) at (p1t -| p3) {$ \colvar_2 < 381 $};
    \node[predicate node, anchor=center] (p4t) at (p1t -| p4) {$ \colvar_3 > 655 $};
    \draw[->] (p1.south) -- (p1t.north);
    \draw[->] (p2.south) -- (p2t.north);
    \draw[->] (p3.south) -- (p3t.north);
    \draw[->] (p4.south) -- (p4t.north);
    \node[preprocessing node] (gdp1) at (p1.south) {$ -69 $ \\ $ \times 1 $};
    \node[preprocessing node] () at (p2.south) {$ -69 $ \\ $ \times 1 $ };
    \node[preprocessing node] () at (p3.south) {$ -69 $ \\ $ \times 1 $};
    \node[preprocessing node] (gdp4) at (p4.south) {$ -25 $ \\$ \times 10 $ };

    \begin{axis}[
        1d histogram plot,
        xmin=0, xmax=500,
        name=p1hist,
        yshift=-\figqueryvspacing,
        at={(p1t.south)},
        anchor=north,
        ]
        \addplot [fill=black!15, draw=none, bar width=81, bar shift=40.5] coordinates {(0,441)};
        \addplot [fill=\figqueryhistcolour, draw=none, bar width=19, bar shift=9.5] coordinates {(81,441)};
        \addplot [fill=\figqueryhistcolour, draw=none, bar width=100, bar shift=50] coordinates {(100,385) (200,291) (300,168) (400,81)};
    \end{axis}
    \begin{axis}[
        1d histogram plot,
        xmin=0, xmax=500,
        name=p2hist,
        yshift=-\figqueryvspacing,
        at={(p2t.south)},
        anchor=north,
        ]
        \addplot [fill=\figqueryhistcolour, draw=none, bar width=100, bar shift=50] coordinates {(0,441) (100,385)};
        \addplot [fill=\figqueryhistcolour, draw=none, bar width=31, bar shift=15.5] coordinates {(200,291)};
        \addplot [fill=black!15, draw=none, bar width=69, bar shift=34.5] coordinates {(231,291)};
        \addplot [fill=black!15, draw=none, bar width=100, bar shift=50] coordinates {(300,168) (400,81)};
    \end{axis}
    \begin{axis}[
        1d histogram plot,
        xmin=0, xmax=500,
        name=p3hist,
        yshift=-\figqueryvspacing,
        at={(p3t.south)},
        anchor=north,
        ]
        \addplot [fill=\figqueryhistcolour, draw=none, bar width=100, bar shift=50] coordinates {(0,441) (100,385) (200,291)};
        \addplot [fill=\figqueryhistcolour, draw=none, bar width=81, bar shift=40.5] coordinates {(300,168)};
        \addplot [fill=black!15, draw=none, bar width=19, bar shift=9.5] coordinates {(381,168)};
        \addplot [fill=black!15, draw=none, bar width=100, bar shift=50] coordinates {(400,81)};
    \end{axis}
    \begin{axis}[
        1d histogram plot,
        xmin=0, xmax=320,
        name=p4hist,
        yshift=-\figqueryvspacing,
        at={(p4t.south)},
        anchor=north,
        ]
        \addplot [fill=black!15, draw=none, bar width=40, bar shift=20] coordinates {(0,84)};
        \addplot [fill=black!15, draw=none, bar width=25, bar shift=12.5] coordinates {(40,112)};
        \addplot [fill=\figqueryhistcolour, draw=none, bar width=15, bar shift=7.5] coordinates {(65,112)};
        \addplot [fill=\figqueryhistcolour, draw=none, bar width=40, bar shift=20] coordinates {(80,195) (120,168) (160,313) (200,223) (240,109) (280,37)};
    \end{axis}
    \draw[->] (p1t.south) -- (p1hist.north);
    \draw[->] (p2t.south) -- (p2hist.north);
    \draw[->] (p3t.south) -- (p3hist.north);
    \draw[->] (p4t.south) -- (p4hist.north);
    \node[anchor=north east, inner sep=0.13em] () at (p1hist.north east) {\tiny $ \bincounts^{\myscale{(2)}} $};
    \node[anchor=north east, inner sep=0.13em] () at (p2hist.north east) {\tiny $ \bincounts^{\myscale{(2)}} $};
    \node[anchor=north east, inner sep=0.13em] () at (p3hist.north east) {\tiny $ \bincounts^{\myscale{(2)}} $};
    \node[anchor=north west, inner sep=0.13em, xshift=0.1em] () at (p4hist.north west) {\tiny $ \bincounts^{\myscale{(3)}} $};

    \node[anchor=north, yshift=-\figqueryvspacing, font={\tiny}] (beta1) at (p1hist.south) {$ \coverages_{1}^{(2)} = \langle 0.19, 1, 1, 1, 1 \rangle $};
    \node[anchor=north, yshift=-\figqueryvspacing, font={\tiny}] (beta2) at (p2hist.south) {$ \coverages_{2}^{(2)} = \langle1, 1, 0.31, 0, 0 \rangle $};
    \node[anchor=north, yshift=-\figqueryvspacing, font={\tiny}] (beta3) at (p3hist.south) {$ \coverages_{3}^{(2)} = \langle1, 1, 1, 0.81, 0 \rangle $};
    \node[anchor=north, yshift=-\figqueryvspacing, font={\tiny}] (beta4) at (p4hist.south) {$ \coverages_{4}^{(3)} = \langle 0, 0.375, 1, \ldots \rangle $};
    \draw[->] (p1hist.south) -- (beta1.north);
    \draw[->] (p2hist.south) -- (beta2.north);
    \draw[->] (p3hist.south) -- (beta3.north);
    \draw[->] (p4hist.south) -- (beta4.north);

    \draw[draw opacity=0] (beta1.south) -- (beta2.south) node[draw opacity=1, midway, anchor=north, yshift=-2*\figqueryvspacing, font={\scriptsize}] (beta12) {$ \coverages_{12}^{(2)} = \langle 0.19, 1, 0.31, 0, 0 \rangle $};
    \draw[->] (beta1.south) -- ([yshift=-1*\figqueryvspacing]beta1.south) -- ([yshift=-1*\figqueryvspacing]beta1.south -| beta12.north) -- (beta12.north);
    \draw[->] (beta2.south) -- ([yshift=-1*\figqueryvspacing]beta2.south) -- ([yshift=-1*\figqueryvspacing]beta2.south -| beta12.north) -- (beta12.north);

    \matrix[matrix anchor=north, yshift=-2*\figqueryvspacing] (weightings) at (beta12.south -| query.south) {
        \node[inner sep=0] (w) {
            $ \weightings^{(1)} = \boldsymbol{1} -
            \left(\boldsymbol{1} - \bincounts^{(12)} \coverages_{12}^{(2)} \hadamarddivision \bincounts^{(1)} \right)
            \hadamardproduct \left( \boldsymbol{1} - \bincounts^{(12)} \coverages_{3}^{(2)} \hadamardproduct \bincounts^{(13)} \coverages_{4}^{(3)} \hadamarddivision \left(\bincounts^{(1)}\right)^2 \right) $ }; \\
    };
    \draw[->] (beta12.south) -- ([yshift=-1*\figqueryvspacing]beta12.south) -- ([yshift=-1*\figqueryvspacing, xshift=-5.7em]beta12.south -| weightings.north) -- ([xshift=-5.7em]weightings.north);
    \draw[->] (beta3.south) -- ([yshift=-3*\figqueryvspacing]beta3.south) -- ([yshift=-3*\figqueryvspacing, xshift=4.2em]beta3.south -| weightings.north) -- ([xshift=4.2em]weightings.north);
    \draw[->] (beta4.south) -- ([yshift=-3*\figqueryvspacing]beta4.south) -- ([yshift=-3*\figqueryvspacing, xshift=9.5em]beta4.south -| weightings.north) -- ([xshift=9.5em]weightings.north);

    \scoped[on background layer] {
        \node[anchor=center, fill=\figqueryhighlightcolour, rounded corners, minimum height=2.2em, minimum width=10.8em] (delayedtransformation) at (beta12) {};

        \node[anchor=center, fill=\figqueryhighlightcolour, rounded corners, minimum height=2.5em, minimum width=15.0em] (andoperator) at ([xshift=8.9em]weightings) {};

        \node[anchor=center, fill=\figqueryhighlightcolour, rounded corners, minimum height=2.5em, minimum width=8em] (transform) at ([xshift=-6.1em]weightings) {};

        \node[anchor=center, fill=\figqueryhighlightcolour, rounded corners, fit={(gdp1) (gdp4)}, yshift=0.4em] (gdpreprocessing) {};  
    };

    \node[font=\tiny, anchor=west, align=left, yshift=-0.35em, xshift=-0.2em] () at (delayedtransformation.east) {Delayed \\ transformation};

    \node[font=\tiny, inner sep=0em, anchor=north, align=center, yshift=-0.7em] () at (transform.south) {Transform coverage to \\ aggregation column};

    \node[font=\tiny, inner sep=0em, anchor=north, yshift=0em] () at (andoperator.south) {AND relation};

    \node[font=\tiny, anchor=west, align=left, yshift=-0.3em] () at (gdpreprocessing.east) {GreedyGD\\pre-process};

\end{tikzpicture}
    \caption{
        Partial query execution for a query aggregating on column 1 with predicates on columns 2 (dist) and 3 (air time), including applying GreedyGD pre-processing, computing coverage for each predicate and bin weightings.
    }
    \label{fig:query_execution}
\end{figure*}
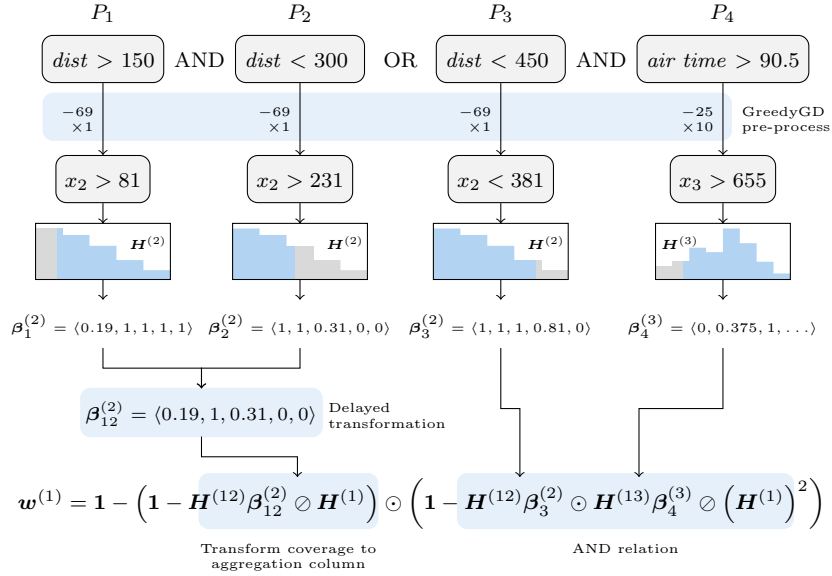

This section explains how AQP tasks are executed using \histname{} executes (top section of \autoref{fig:system}), which is illustrated in \autoref{fig:query_execution}.

\subsection{SQL parsing}

SQL queries are parsed by applying GreedyGD pre-processing to predicate literals so that they are in the same domain as the compressed data on which \histname{} is built.
For example, in \autoref{fig:query_execution}, the \textit{dist} column's minimum value of 69 is subtracted from predicates 1--3, while the \textit{air time} column's minimum value of 25 is subtracted from predicate 4.
Predicate 4 is also multiplied by 10 to convert from floating point to integer.

\subsection{Coverage}

Given a condition $ \predicate $ that applies to column $ j $, we define \textit{coverage}, $ \coverages^{(j)} $, as the $ \nbins^{(j)} \times 1 $ vector whose $ \binidx $th element is the probability that a point in the $ \binidx $th bin of the one-dimensional histogram for column $ j $ satisfies $ \predicate $.
That is,
\begin{equation}
    \coverage^{(j)}_{\binidx} = \prob \Big( \predicate \setsep \edge_{\binidx}^{(j)} \le x < \edge_{\binidx + 1}^{(j)} \Big)
\end{equation}
where $ \prob $ is the probability operator.
This is estimated as follows for equality conditions (inverse for inequality):
\begin{equation}
    \coverage^{(j)}_{\binidx} =
    \begin{cases}
        0, & \text{condition value outside bin} \\
        1 / \binuniquecount_{\binidx}^{(j)}, & \text{otherwise,}
    \end{cases}
\end{equation}
and for all other condition types (i.e., $ \le $, $ < $, $ \ge $, $ > $) as:
\begin{equation}
    \coverage^{(j)}_{\binidx} =
    \begin{cases}
        0, & \binmin{t}{(j)} \text{ and } \binmax{t}{(j)} \text{ fail } \predicate \\
        1, & \binmin{t}{(j)} \text{ and } \binmax{t}{(j)} \text{ satisfy } \predicate \\
        0.5, & \text{one of } \binmin{t}{(j)}, \binmax{t}{(j)} \text{ satisfy } \predicate \text{ \& } \binuniquecount_{\binidx}^{(j)} = 2 \\
        f_{\binidx}(\predicate), & \text{otherwise,}
    \end{cases}
\end{equation}
where $ f_{\binidx}(\predicate) $ is the fraction of the bin width ($ \binwidth_{\binidx} $) that satisfies $ \predicate $.

Coverage is estimated separately for each predicate condition, as illustrated in \autoref{fig:query_execution} with $ \coverages_{1} $ to $ \coverages_{4} $ corresponding to predicates $ \predicate_1 $ to $ \predicate_4 $.
Groups of conditions that apply to the same column are consolidated.
By `group' we mean any group of two or more conditions that are directly connected by a single AND or OR operator.
For example, in \autoref{fig:query_execution}, $ \coverages_{1} $ and $ \coverages_{2} $ are consolidated into $ \coverages_{12} $ since they both apply to the \textit{dist} column and are directly connected by an AND operator.
However, despite $ \predicate_3 $ also applying to the same column, it is not consolidated since it must be first combined with $ \predicate_4 $ due to operator precedence (AND before OR).
We refer to this process as \textit{delayed transformation} because we delay transforming coverages into weightings in the aggregation dimension (\autoref{sec:query_execution:weightings}) to consolidate same-column conditions.

Coverage bounds are also computed and used to bound the query result.
The source of uncertainty here is the unknown distribution of data points within partially covered bins.
Applying similar logic to \autoref{sec:histogram:bin_centres}, we can show that, for bins that do not pass the hypothesis test in \algoIsUniform{} (i.e., bins with $ \bincount < \minpts $), coverage bounds occur when only one data point or all but one data points satisfy $ \predicate $.
However, for bins with at least $ \minpts $ points, tighter bounds can be derived by considering the \textit{partial bin count}, which is defined as the number of data points in a subset of the bin's sub-bins.
\autoref{theorem:bin_partial_count} provides the key result for these bounds.

\begin{theorem}
    \label{theorem:bin_partial_count}
    Consider a histogram bin with count $ h $ and $ u $ unique values.
    Assume that this bin satisfies the hypothesis test in \textup{\algoIsUniform{}} with $ s = \big\lceil(2u)^{1/3}\big\rceil $ sub-bins and critical value $ \chisq{\testsignificance} $.
    The minimum and maximum partial bin count over $ \nsubbinscovered \le s $ sub-bins is then:
    \begin{equation}
        \partialbincount{\nsubbinscovered}{\nsubbins}^{\pm} = \frac{\bincount \nsubbinscovered}{\nsubbins} \pm \frac{\bincount \nsubbinscovered}{\nsubbins} \sqrt{\frac{\chisq{\testsignificance} (\nsubbins - \nsubbinscovered)} {\bincount \nsubbinscovered}} \text{.}
    \end{equation}
\end{theorem}
\begin{proof}
    Use Lagrange Multipliers to optimise the partial count $ \sum_{\subbinidx=1}^{\nsubbinscovered} \subbincount_{\subbinidx} $, where $ \subbincount_{\subbinidx} $ is the count for the $ \subbinidx $th sub-bin, given constraints:
    \begin{equation} \label{eq:theorem2:constraints}
        \sum_{\subbinidx=1}^{\nsubbins} \subbincount_{\subbinidx} = h
        \text{ and }
        \chisq{\testsignificance} = \sum_{\subbinidx=1}^{\nsubbins}\frac{(\subbincount_{\subbinidx} - \bincount/\nsubbins)^2}{\bincount/\nsubbins}.
    \end{equation}
    The Lagrangian is then:
    \begin{align}
        \mathcal{L}(\subbincount_{\subbinidx}, \lambda_1, \lambda_2) = \sum_{\subbinidx=1}^{\nsubbinscovered} \subbincount_{\subbinidx} &+ \lambda_1 \left(\bincount - \sum_{\subbinidx=1}^{\nsubbins} \subbincount_{\subbinidx}\right) \nonumber \\[-0.2em]
        & + \lambda_2 \left(\chisq{\testsignificance} - \sum_{\subbinidx=1}^{\nsubbins}\frac{(\subbincount_{\subbinidx} - h/\nsubbins)^2}{\bincount/\nsubbins}\right) \text{.}
    \end{align}
    Setting the partial derivative of $ \mathcal{L} $ with respect to $ \subbincount_{\subbinidx} $ to zero gives
    \begin{align}
        \subbincount_{\subbinidx} =
        \begin{cases}
            \frac{h}{\nsubbins} + \frac{\bincount(1-\lambda_1)}{2\nsubbins\lambda_2}, & \subbinidx < \nsubbinscovered \\
            \frac{h}{\nsubbins} - \frac{\bincount \lambda_1}{2\nsubbins\lambda_2}, & \subbinidx \ge \nsubbinscovered
        \end{cases}
        \label{eq:sub_bin_count_initial_partial}
    \end{align}
    Substituting this into the two constraints in \autoref{eq:theorem2:constraints} and solving for $ \lambda_1 $ and $ \lambda_2 $ gives the following result for $ \subbincount_{\subbinidx} $:
    \begin{align}
        \subbincount_{\subbinidx} =
        \begin{cases}
            \frac{\bincount}{\nsubbins} \pm \frac{\bincount}{\nsubbins} \sqrt{\chisq{\testsignificance} (\nsubbins - \nsubbinscovered) \, / \, \bincount \nsubbinscovered}, & \subbinidx < \nsubbinscovered \text{,} \\[0.3em]
            \frac{\bincount}{\nsubbins} \mp \frac{\bincount}{\nsubbins} \sqrt{\chisq{\testsignificance} \nsubbinscovered \, / \, \bincount (\nsubbins - \nsubbinscovered)} & \subbinidx \ge \nsubbinscovered \text{.} \\
        \end{cases}
        \label{eq:sub_bin_count_final}
    \end{align}
    Multiplying the solution for $ \subbinidx < \nsubbinscovered $ by $ \nsubbinscovered $ gives the desired result.
\end{proof}

Given \autoref{theorem:bin_partial_count}, coverage bounds $ \coverages^{-} $ and $ \coverages^{+} $ can be obtained by dividing the partial bin count by the total bin count with the appropriate number of sub-bins.
That is,
\begin{align}
    \coverage_{\binidx}^{-(j)} &=
    \begin{cases}
        \coverage^{(j)}_{\binidx}, & \coverage_{\binidx}^{(j)} \in \{0, 1\} \\
        1 \, / \, \bincount_{\binidx}^{(j)}, & \coverage_{\binidx}^{(j)} \notin \{0, 1\} \text{ \& } \bincount_{\binidx}^{(j)} < \minpts \\
        \frac{\nsubbinsfullycovered}{\nsubbins} - \frac{\nsubbinsfullycovered}{\nsubbins} \sqrt{ \chisq{\testsignificance} (\nsubbins - \nsubbinsfullycovered) / \bincount \nsubbinsfullycovered}, & \text{otherwise,}
    \end{cases} \label{eq:coverage_low} \\[0.3em]
    \coverage_{\binidx}^{+(j)} &=
    \begin{cases}
        \coverage^{(j)}_{\binidx}, & \coverage_{\binidx}^{(j)} \in \{0, 1\} \\
        1 - 1 \, / \, \bincount_{\binidx}^{(j)}, & \coverage_{\binidx}^{(j)} \notin \{0, 1\} \text{ \& } \bincount_{\binidx}^{(j)} < \minpts \\
        \frac{\nsubbinspartiallycovered}{\nsubbins} + \frac{\nsubbinspartiallycovered}{\nsubbins} \sqrt{ \chisq{\testsignificance} (\nsubbins - \nsubbinspartiallycovered) / \bincount \nsubbinspartiallycovered}, & \text{otherwise,}
    \end{cases} \label{eq:coverage_high}
\end{align}
where $ \nsubbins = \big\lceil\big(2 u_{\binidx}^{(j)}\big)^{1/3} \big\rceil $ is the number of sub-bins, $ \nsubbinsfullycovered = \floor{\coverage_{\binidx}^{(j)} \nsubbins} $ is the number of sub-bins that are fully covered by the predicate condition and $ \nsubbinspartiallycovered = \ceil{\coverage_{\binidx}^{(j)} \nsubbins} $ is the number of fully or partially covered sub-bins.

\begin{table*}[!t]
    \centering
    \caption{Aggregation functions}
    \hspace{\tablecaptiongapskrinkage}
    \scriptsize
    \newcommand{\aggtablespacing}{0.5em}
\begin{tabular}{llll}
    \toprule
    \textbf{Aggregation} & \textbf{Estimate}        & \textbf{Lower bound} & \textbf{Upper bound} \\ \midrule
    COUNT                & $ \lonenorm{\weightings} \ / \ \sr $
&
$ \lonenorm{\weightings^{-}} \ / \ \sr $
&
$ \lonenorm{\weightings^{+}} \ / \ \sr $
  \\ [\aggtablespacing]
    SUM                  & $ \weightings \dotprod \binmidpoints \ / \ \sr $
&
$ \weightings^{-} \dotprod \binmidpoints^{-} \ / \ \sr $
&
$ \weightings^{+} \dotprod \binmidpoints^{+} \ / \ \sr $
    \\ [\aggtablespacing]
    AVG                  & $ \weightings \dotprod \binmidpoints \ / \ \lonenorm{\weightings} $
&
$ \min\limits_{\{\weightings^{-}, \ \weightings^{+}\}} \left\{ \weightingsplaceholder \dotprod \binmidpoints^{-} \ / \  \lonenorm{\weightingsplaceholder}  \right\} $
&
$ \max\limits_{\{\weightings^{-}, \ \weightings^{+}\}} \left\{ \weightingsplaceholder \dotprod \binmidpoints^{+} \ / \  \lonenorm{\weightingsplaceholder}  \right\} $
\vspace{0.6em}
    \\ [\aggtablespacing]
    MIN                  & $
\begin{cases}
    \binmax{\binkey}{}, & \text{1-d, } \binuniquecount_{\binkey} = 2 \text{ \& } \weighting_{\binkey} < \frac{\bincount_{\binkey}}{2} \\
    \binmin{\binkey}{}, & \text{otherwise.} \\
\end{cases}
$
&
$
\begin{cases}
    \binmax{\binkey}{}, & \text{1-d, } \binuniquecount_{\binkey} = 2 \text{ \& } \weighting_{\binkey}^{+} < \frac{\bincount_{\binkey}}{5} \\
    \binmin{\binkey}{}, & \text{otherwise.} \\
\end{cases}
$
&
$
\begin{cases}
    \binmax{\binkey}{} - \nsubbinsfullycovered \subbinwidth_{\binkey}, & \text{1-d, } \binuniquecount_{\binkey} > 2 \text{ \& } \bincount_{\binkey} > \minpts \\
    \binmax{\binkey}{},                                & \text{otherwise,}
\end{cases}
$
\vspace{0.7em}
    \\ [\aggtablespacing]
    MAX                  & $
\begin{cases}
    \binmin{\binkey}{}, & \text{1-d, } \binuniquecount_{\binkey} = 2 \text{ \& } \weighting_{\binkey} < \frac{\bincount_{\binkey}}{2} \\
    \binmax{\binkey}{}, & \text{otherwise.} \\
\end{cases}
$
&
$
\begin{cases}
    \binmin{\binkey}{} + \nsubbinsfullycovered \subbinwidth_{\binkey}, & \text{1-d, } \binuniquecount_{\binkey} > 2 \text{ \& } h_{\binkey} > \minpts \\
    \binmin{\binkey}{},                                               & \text{otherwise,}
\end{cases}
$
&
$
\begin{cases}
    \binmin{\binkey}{}, & \text{1-d, } \binuniquecount_{\binkey} = 2 \text{ \& } \weighting_{\binkey}^{+} < \frac{\bincount_{\binkey}}{5} \\
    \binmax{\binkey}{}, & \text{otherwise.} \\
\end{cases}
$
\vspace{0.6em}
    \\ [\aggtablespacing]
    MEDIAN               & $
\begin{cases}
    \binmin{\binkey}{}, & \binuniquecount_{\binkey} = 2 \ \& \ \fractomedian < 0.5 \\
    \binmax{\binkey}{}, & \binuniquecount_{\binkey} = 2 \ \& \ \fractomedian \ge 0.5 \\
    \binmin{\binkey}{} + \binwidth_{\binkey} \fractomedian, & \text{otherwise,}
\end{cases}
$
&
$ \binmin{\binkey}{} $
&
$ \binmax{\binkey}{} $
\vspace{-0.3em}
 \\ [\aggtablespacing]
    VAR                  & $
\weightings \cdot \binmidpoints^2 \ / \ \lonenorm{\weightings} - \left( \weightings \cdot \binmidpoints \ / \ \lonenorm{\weightings} \right)^2 \text{.}
$
&
{$
\begin{aligned}
    \min\limits_{\{\weightings^{-}, \ \weightings^{+}\}}
    \Big\{ &
    \weightingsplaceholder \cdot (\varboundsxs^{-})^2 \,/\, \lonenorm{\weightingsplaceholder} \\[-1.0em]
    & - \left( \weightingsplaceholder \cdot \varboundsxs^{-} \,/\, \lonenorm{\weightingsplaceholder} \right) ^2
    \Big\}
\end{aligned}
$}
&
{$
\begin{aligned}
    \max\limits_{\{\weightings^{-}, \ \weightings^{+}\}}
    \Big\{ &
    \weightingsplaceholder \cdot (\varboundsxs^{+})^2 \,/\, \lonenorm{\weightingsplaceholder} \\[-1.0em]
    & - \left( \weightingsplaceholder \cdot \varboundsxs^{+} \,/\, \lonenorm{\weightingsplaceholder} \right) ^2
    \Big\}
\end{aligned}
$}
    \\ [\aggtablespacing] \bottomrule
\end{tabular}

    \label{tab:aggregations}
\end{table*}

\subsection{Weightings}
\label{sec:query_execution:weightings}

Given query predicate $ \predicates $ containing conditions $ \predicate_1, \ldots, \predicate_{\npredicates} $ and aggregation column $ i $, we define bin \textit{weightings}, $ \weightings^{(i)} $, as the $ \nbins^{(i)} \times 1 $ vector whose $ \binidx $th element is the estimated number of points in the $ \binidx $th bin of the one-dimensional histogram for column $ i $ that satisfy $ \predicates $.
This can be expressed as follows:
\begin{equation}
    \weighting^{(i)}_{\binidx} = \bincount_{\binidx}^{(i)} \prob \Big( \predicates \setsep \edge_{\binidx}^{(i)} \le x < \edge_{\binidx + 1}^{(i)} \Big).
\end{equation}
For predicates that are the intersection (AND) of multiple conditions $ \predicate_1, \ldots, \predicate_{\npredicates} $ and assuming conditional independence between predicate conditions, this can be expressed as follows:
\begin{align}
    \weighting^{(i)}_{\binidx}
    &= \bincount_{\binidx}^{(i)} \prob \left( \predicate_1 \cap \predicate_2 \cap \ldots \cap \predicate_{\npredicates} \setsep \edge_{\binidx}^{(i)} \le \colvar < \edge_{\binidx + 1}^{(i)} \right) \nonumber \\
    &= \bincount_{\binidx}^{(i)} \prod_{\predicateidx=1}^{\npredicates} \prob \left( \predicate_{\predicateidx} \setsep \edge_{\binidx}^{(i)} \le \colvar < \edge_{\binidx + 1}^{(i)} \right).
\end{align}
Likewise, for predicates that are the union (OR) of multiple conditions and again assuming conditional independence between predicate conditions, weightings can be expressed as follows:
\begin{align}
    \weighting^{(i)}_{\binidx}
    &= \bincount_{\binidx}^{(i)} \prob \left( \predicate_1 \cup \predicate_2 \cup \ldots \cup \predicate_{\npredicates} \setsep \edge_{\binidx}^{(i)} \le \colvar < \edge_{\binidx + 1}^{(i)} \right) \nonumber\\
    &= \bincount_{\binidx}^{(i)} \left(1 - \prob \left( \bar{\predicate}_1 \cap \bar{\predicate}_2 \cap \ldots \cap \bar{\predicate}_{\npredicates} \setsep \edge_{\binidx}^{(i)} \le \colvar < \edge_{\binidx + 1}^{(i)} \right) \right) \nonumber\\
    &= \bincount_{\binidx}^{(i)} \left(1 - \prod_{\predicateidx=1}^{\npredicates} \left( 1 - \prob \left( \predicate_{\predicateidx} \setsep \edge_{\binidx}^{(i)} \le \colvar < \edge_{\binidx + 1}^{(i)} \right) \right)\right).
\end{align}
We then observe that
\begin{equation}
    \prob \left( \predicate_{\predicateidx} \setsep \edge_{\binidx}^{(i)} \le \colvar < \edge_{\binidx + 1}^{(i)} \right)
    = \frac{1}{\bincount^{(i)}_{\binidx}} \left[ \bincounts^{(ij)} \coverages^{(j)} \right]_{\binidx},
\end{equation}
which allows us to express weightings as follows:
\begin{align} \label{eq:weightings}
    \weightings^{(i)} =
    \begin{cases}
        \prod\limits_{\predicateidx=1}^{\npredicates} \bincounts^{(ij)} \coverages^{(j)} \hadamarddivision \bincounts^{(i)}, &\text{intersection} \\
        1 - \prod\limits_{\predicateidx=1}^{\npredicates} \left(1 - \bincounts^{(ij)} \coverages^{(j)} \hadamarddivision \bincounts^{(i)} \right), &\text{union} \\
    \end{cases}
\end{align}
where the product ($ \Pi $) is element-wise (Hadamard).
By combining these results, arbitrary combinations of AND and OR relations can be processed.
Note that due to assuming conditional independence, \autoref{eq:weightings} may not be reliable for highly correlated data.
This is especially true for multiple conditions on the same column, which are clearly not (conditionally) independent, and thus why we perform delayed transformation to consolidate such conditions prior to computing bin weightings.

Weightings bounds, $ \weightings^{-} $ and $ \weightings^{+} $, are computed using \autoref{eq:weightings}, using low and high coverage bounds, respectively.
If \histname{} is constructed from a data sample, these bounds are widened to account for additional uncertainty due to sampling.
That is, the bounds are replaced by their outer two-sided 98-percentile confidence interval bounds.
Variance is estimated according to the Binomial distribution as $ \coverage_{\binidx} (1 - \coverage_{\binidx}) $ where $ \coverage_{\binidx} = \weighting_{\binidx} / \bincount_{\binidx} $.
Including compensation for finite population size, the weightings bounds are updated as follows:
\begin{align}
    \weighting_{\binidx}^{\pm} \leftarrow \weighting_{\binidx}^{\pm} \pm \ciweight \sqrt{ \coverage_{\binidx}^{\pm} \left(1 - \coverage_{\binidx}^{\pm} \right) \frac{\nrows - \nsamples}{\nrows - 1} },
\end{align}
where $ \ciweight $ is the value of the standard normal density function corresponding to a two-sided 98-percentile confidence interval.

\subsection{Aggregation}
\label{sec:query_execution:aggregation}

This section describes the mathematical formulations for query aggregation, which are listed in \autoref{tab:aggregations}, along with corresponding bounds.
Currently, seven aggregation functions are supported, namely COUNT, SUM, AVG, MIN, MAX, MEDIAN and VAR.

\subsubsection{COUNT}

This can be estimated by summing the weightings vector and then dividing by the sampling ratio, $ \sr $.

\subsubsection{SUM}

This can be estimated as the weighted sum of bin midpoints, $ \binmidpoints $, using the estimated bin weightings, $ \weightings $.
As with COUNT, the result must also be scaled by the sampling ratio, $ \sr $.

\subsubsection{AVG}

This can be estimated as the weighted mean of bin midpoints, $ \binmidpoints $, using the estimated bin weightings, $ \weightings $.
The bounds follow the same formulation, but naturally use the appropriate bounds for bin midpoints.
Additionally, each bound is estimated using both bin weightings extrema and either the minimum or maximum is returned.
For example, the lower bound is the minimum over $ \weightings^{-} $ and $ \weightings^{+} $ where $ \weightingsplaceholder $ is a placeholder variable.

\subsubsection{MIN}

Estimating MIN requires identifying the index, $ \binkey $, of the first bin with non-zero weighting, i.e.,
\begin{align}
    \binkey = \min_{\binidx}\{\binidx \text{, s.t. } \weighting_{\binidx} > 0\} .
\end{align}
In most cases, the estimate is the minimum value for this bin, $ \binmin{\binkey}{} $.
For the special case where the query involves only a single column (for aggregation and all predicates), there are only two unique values in the $ \binkey $th bin (i.e., $ \binuniquecount_{\binkey} = 2 $) and the bin coverage is less than half (i.e., $ \weighting_{\binkey} < \bincount_{\binkey} / 2 $), then the bin maximum, $ \binmax{\binkey}{} $, is a better estimate.

The lower bound follows the same formulation, except that $ \binkey $ must be determined from the bin weightings upper bound (to include the maximum range), i.e.,
\begin{align}
    \binkey = \min_{\binidx}\{\binidx \text{, s.t. } \weighting^{+}_{\binidx} > 0\} \text{\hspace{2em} (lower bound)}.
\end{align}
Conversely, the bin corresponding to the upper bound is identified using the bin weightings lower bound.
A higher threshold is used to ensure a higher likelihood of non-zero true bin coverage, i.e.,
\begin{align}
    \binkey = \min_{\binidx}\{\binidx \text{, s.t. } \weighting^{-}_{\binidx} > 1/2\} \text{\hspace{2em} (upper bound)}.
\end{align}
A tighter upper bound can be obtained for queries involving a single column in certain cases by considering the sub-bins.
That is, if the relevant bin passed the uniformity hypothesis test, and hence $ \bincount \ge \minpts $, then we may assume that data is uniformly distributed across the sub-bins.
Thus, we can compute the number of sub-bins that are fully covered as $ \nsubbinsfullycovered = \floor{ \nsubbins \weighting_{\binkey}^{-} / \bincount_{\binkey} } $ and reduce the upper bound by this number of sub-bin widths, $ \subbinwidth_{\binkey} $.

\subsubsection{MAX}

This is the inverse of MIN.
That is, we identify the index, $ \binkey $, of the \textit{last} bin with non-zero weighting, i.e.,
\begin{align}
    \binkey = \max_{\binidx}\{\binidx \text{, s.t. } \weighting_{\binidx} > 0\},
\end{align}
and use this to determine the estimates.
Similar, but inverse, formulations apply for the lower and upper bounds.

\subsubsection{MEDIAN}

This is estimated by first identifying the index, $ \binkey $, of the bin that contains the median.
That is,
\begin{equation}
    \binkey = \min_{\binidx} \left\{ \binidx, \text{ s.t. } \sum_{\tau=0}^{\binidx} \weighting_{\tau} \ge \frac{1}{2} \lonenorm{\weightings} \right\} .
\end{equation}
The estimate is then the bin minimum, $ \binmin{\binkey}{} $ plus the bin width, $ \binwidth_{\binkey} = \binmax{\binkey}{} - \binmin{\binkey}{} $, multiplied by the fraction of the bin weighting that is below the median, $ \fractomedian $, which is calculated as follows:
\begin{equation}
    \fractomedian = \frac{1}{\weighting_{\binkey}} \left( \frac{1}{2} \norm{\weightings}_{1} - \sum\limits_{\binidx=0}^{\binkey - 1} \weighting_{\binidx}  \right)
\end{equation}
If the bin contains only two unique values (i.e., $ \binuniquecount_{\binkey} = 2 $), the bin minimum or maximum is returned instead, depending on the value of $ \fractomedian $.

MEDIAN bounds require identifying the minimum and maximum bin indices that could correspond to the median, i.e.,
\begin{align}
    \binkey &=
    \min_{\{\weightings^{-}, \ \weightings^{+}\}} \left\{
        \min_{\binidx} \left\{ \binidx, \text{ s.t. } \sum_{\tau=0}^{\binidx} \weightingplaceholder_{\tau} \ge \frac{1}{2} \right\}
    \right\} \text{\hspace{0.5em} (lower bound)} \\
    \binkey &=
    \max_{\{\weightings^{-}, \ \weightings^{+}\}} \left\{
        \min_{\binidx} \left\{ \binidx, \text{ s.t. } \sum_{\tau=0}^{\binidx} \weightingplaceholder_{\tau} \ge \frac{1}{2} \right\}
    \right\} \text{\hspace{0.5em} (upper bound)}
\end{align}

\subsubsection{VAR}

This is estimated as the difference between the weighted mean square value and the square of the estimated mean.
The idea for the bounds is to assume that all points within a bin are either as far from the mean as possible (upper bound) or as close to the mean as possible (lower bound).
This is expressed using $ \varboundsxs^{-} $ and $ \varboundsxs^{+} $, which represent the (assumed) location of points within each bin for the purpose of estimating the bounds and are defined as follows:
\begin{align}
    \varboundsx^{-}_{\binidx} &=
    \begin{cases}
        \binmax{\binidx}{}, & \binmax{\binidx}{} < AVG \\
        \binmin{\binidx}{}, & \binmin{\binidx}{} > AVG \\
        AVG, & \text{otherwise,}
    \end{cases} \\
    \varboundsx^{+}_{\binidx} &=
    \begin{cases}
        \binmin{\binidx}{}, & | AVG - \binmin{\binidx}{} | > | \binmax{\binidx}{} - AVG | \\
        \binmax{\binidx}{}, & \text{otherwise.} \\
    \end{cases}
\end{align}
where $ AVG $ is the estimated mean.
Thus, for the lower bound, bins are represented by whichever of the bin minimum or maximum is closest to the estimated mean, while the bin that straddles the mean is represented by the mean itself.
Conversely, for the upper bound, bins are represented by whichever of the bin minimum or maximum is furthest from the mean.
As with AVG and MEDIAN queries, the bounds are evaluated for both weightings extrema and the appropriate min/max is returned.

	\section{Performance Evaluation}
\label{sec:evaluation}

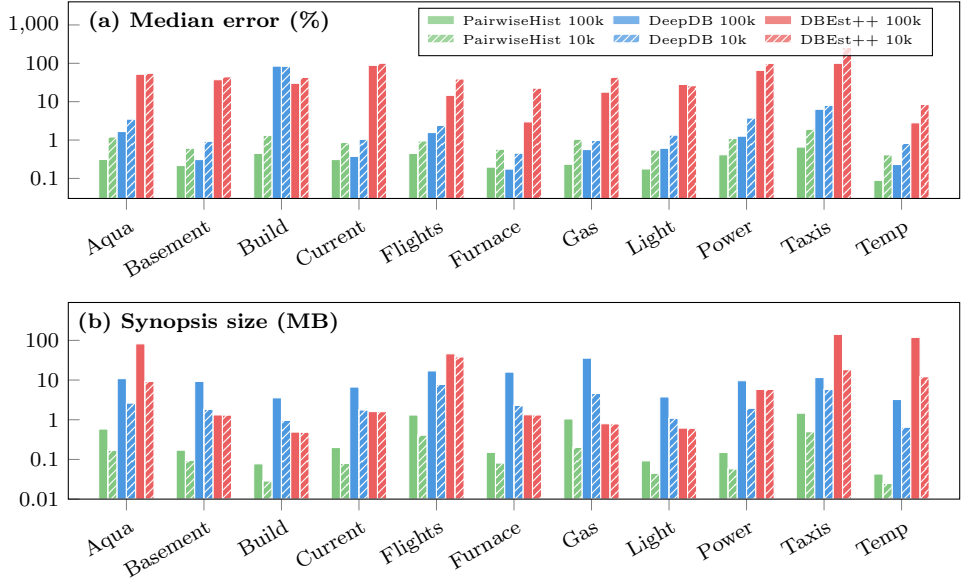
\begin{figure*}[!t]
    \centering
    \begin{tikzpicture}
    \begin{semilogyaxis}[
        dataset grouped bar plot,
        ymin = 3e-2, ymax=4e3,
        ytick = {0.1, 1, 10, 100, 1000},
        at = {(0,0)},
        name = accuracy,
        ]
        \foreach \i in {1,2,3,4,5,6} {
            \addplot+ [] table[x index=0, y index=\i] {\datarealdatasetsacccuracy};
        }
        \legend{\histname{} 100k, \histname{} 10k, DeepDB 100k, DeepDB 10k, DBEst++ 100k, DBEst++ 10k}
    \end{semilogyaxis}

    \begin{semilogyaxis}[
        dataset grouped bar plot,
        ymin = 1e-2, ymax=9e2,
        ytick = {0.01, 0.1, 1, 10, 100},
        at = {(0,0)},
        yshift = -4cm,
        name = storage,
        ]
        \foreach \i in {1,2,3,4,5,6} {
            \addplot+ [] table[x index=0, y index=\i] {\datarealdatasetsstorage};
        }
    \end{semilogyaxis}

    \node[anchor=north west, font=\footnotesize, xshift=0.4em] () at (accuracy.north west) {\textbf{(a) Median error (\%)}};
    \node[anchor=north west, font=\footnotesize] () at (storage.north west) {\textbf{(b) Synopsis size (MB)}};
\end{tikzpicture}
    \caption{
        Error performance and storage requirements across 11 real-world datasets.
    }
    \label{fig:real_datasets_results}
\end{figure*}

To evaluate \histname{}, 11 real-world datasets were used, which are summarised in \autoref{tab:datasets}.
The \textit{Basement}, \textit{Current} and \textit{Furnace} datasets~\citep{Makonin_2016,Makonin_2016a} contain electrical meter data for different areas of a house.
\textit{Gas}~\citep{Huerta_2016}, \textit{Light}~\citep{Prasad_2020}, \textit{Power}~\citep{Hebrail_2012} and \textit{Temp}~\citep{Porter_2021} contain multifaceted IoT sensor data from a single source, while \textit{Aqua}~\citep{Udanor_2022} and \textit{Build}~\citep{Hong_2017} contain IoT sensor data combined from multiple sources (aquaponics ponds and building rooms, respectively) with several data columns each and a shared timestamp column.
\textit{Aqua} and \textit{Build} thus contain many null values due to asynchronous sampling.
\textit{Flights}~\cite{UDT_2016} and \textit{Taxi}~\citep{CityChicago_2022c} contain records of individual trips and include several categorical fields, as well as missing values.
\textit{Flights} is commonly used in AQP literature (e.g.~\cite{Ma_2021,Hilprecht_2020,Kulessa_2018,Eichmann_2020,Sheoran_2022,Garg_2023}), while \textit{Power} is used occasionally (e.g.~\cite{Zhang_2021a,Zhang_2021}).
Overall, these datasets encompass a wide variety of data types, dimensionality, sizes and both sensor and non-sensor data.


\begin{table*}[!t]
    \centering
    \caption{
        Datasets used for evaluation
    }
    \hspace{\tablecaptiongapskrinkage}
    \footnotesize
    \begin{tabular}{llrrr}
    \toprule
    \multicolumn{2}{l}{\textbf{Dataset}}                                       & \textbf{Rows} & \textbf{Columns} &                                                       \textbf{Size} (\unit{MB}) \\ \midrule
    Aqua          & Aquaponics sensors~\citep{Udanor_2022}                     &      913\,465 &               13 &                                                                            66.7 \\
    Basement      & Basement power~\citep{Makonin_2016,Makonin_2016a}          &   1\,051\,200 &               12 &                                                                            50.5 \\
    Build         & Smart building systems~\citep{Hong_2017}                   &  14\,381\,639 &                7 &                                                                           402.7 \\
    Current       & Electric meters current~\citep{Makonin_2016,Makonin_2016a} &   1\,051\,200 &               24 &                                                                           100.9 \\
    Flights$ ^* $ & Flight delays \& cancellations~\cite{UDT_2016}             &   5\,819\,079 &               32 &                                                                           756.5 \\
    Furnace       & Furnace power~\citep{Makonin_2016,Makonin_2016a}           &   1\,051\,200 &               12 &                                                                            50.5 \\
    Gas           & Home gas sensor~\citep{Huerta_2016}                        &      928\,991 &               12 &                                                                            44.6 \\
    Light         & IoT light detection~\citep{Prasad_2020}                    &      405\,184 &                9 &                                                                            19.9 \\
    Power         & Home power consumption~\citep{Hebrail_2012}                &   2\,049\,280 &               10 &                                                                            82.0 \\
    Taxis         & Chicago taxi trips 2020~\citep{CityChicago_2022c}          &   3\,889\,032 &               23 &                                                                        1\,753.9 \\
    Temp          & Temperature sensor~\citep{Porter_2021}                     &  10\,553\,597 &                5 &                                                                           369.4 \\ \bottomrule
    \multicolumn{5}{p{0.6\textwidth}}{\scriptsize $ ^* $Other works typically use only 12 columns from the \textit{Flights} dataset (e.g.~\citep{Hilprecht_2020}). However, we use all 32 columns.}
\end{tabular}

    \label{tab:datasets}
\end{table*}


We implemented \histname{} in Python~3.11 and compared it to state-of-the-art AQP techniques DeepDB~\citep{Hilprecht_2020} and DBEst++~\citep{Ma_2021} using their corresponding Python implementations~\cite{Hilprecht_2020a,Ma_2022}.
These techniques were chosen for comparison due to their leading performance (see \autoref{tab:stateoftheart}) and code availability.


Our evaluation is divided into two parts: initial experiments on the original datasets and comprehensive experiments on scaled-up versions of the \textit{Power} and \textit{Flights} datasets.
For the initial experiments, we randomly generated 100 single-predicate queries for each dataset with aggregation functions COUNT, SUM and AVG and minimum selectivity of 10$ ^{-5} $.
For the scaled-up experiments, IDEBench~\citep{Eichmann_2020} was used to scale the datasets up to one billion rows, resulting in sizes of \qty{40}{GB} and \qty{130}{GB}, respectively.
We then randomly generated 445 and 427 test queries for \textit{Power} and \textit{Flights}, respectively, including all seven aggregation functions supported by \histname{}, 1--5 predicate conditions and minimum selectivity of 10$ ^{-6} $.
Due to aforementioned limitations of DeepDB and DBEst++ (\autoref{sec:related_work}), DeepDB supports only 146 queries (80 for \textit{Power} and 66 for \textit{Flights}), while DBEst++ supports just 86 queries (41 for \textit{Power} and 45 for \textit{Flights}).
Also, since DBEst++ requires many models to support different query templates (see \autoref{sec:related_work}), we include all DBEst++ models required to support the same queries as \histname{} when comparing synopsis size.


Our experimental setup consisted of an Intel(R) Xeon(R) Gold 6130 \qty{2.10}{GHz} CPU with \qty{24}{GB} RAM available during synopsis construction and \qty{6}{GB} during query execution.
Default parameters were used for DeepDB and DBEst++, as well as GreedyGD.
All experiments were performed with $ \minpts $ set to 1\% of $ \nsamples $ (e.g. $ \minpts=10^3 $ for $ \nsamples = 10^5 $) and $ \testsignificance $ set to 0.001.



\subsection{Initial experiments}

The median query error and synopsis size for each dataset is shown in \autoref{fig:real_datasets_results} for \histname{}, DeepDB and DBEst++ with 100k and 10k samples.
As can be seen in \autoref{fig:real_datasets_results}(a), \histname{} has the lowest error on 10 out of 11 datasets.
Indeed, even with a mere 10k samples, \histname{} outperforms DeepDB with 100k samples on 6 out of 11 datasets.
Overall, for 100k samples, \histname{} has a median error of just 0.28\%, compared to 0.73\% for DeepDB and 28.9\% for DBEst++.
In terms of synopsis size, \histname{} is typically 1--2 orders of magnitude smaller.
For 100k samples, the mean size for \histname{} is just \qty{0.48}{MB}, compared to \qty{11.5}{MB} for DeepDB and \qty{36.3}{MB} for DBEst++.
That is, \histname{} is at least 2.6$ \times $ as accurate while requiring 24$ \times $ less storage.

\subsection{Parameter sensitivity}
\label{sec:evaluation:parameters}

\begin{figure}[!t]
    \vspace{-1.0em}
    \centering
    \subfloat{\begin{tikzpicture}[trim axis right]
    \begin{axis}[
        parameter sensitivity plots,
        title = {\textbf{(a) Median error (\%)}},
        ytick = {0, 1, 2, 3, 4, 5},
        ymin=0, ymax=4,
        cycle list={
            {colour1!70, mark options={draw=white,solid,thick,scale=1.15,opacity=1}, mark=*},
            {colour2!50, mark options={draw=white,solid,thick,scale=1.0,opacity=1}, mark=square*, densely dotted},
            {colour2!70, mark options={draw=white,solid,thick,scale=1.0,opacity=1}, mark=square*},
            {colour2!30, mark options={draw=white,solid,thick,scale=1.0,opacity=1}, mark=square*, densely dashed},
        },
        clip mode=individual,
        ]
        \addplot+[]  
        coordinates {
            (1000, 0.445)
            (2000, 0.446)
            (3000, 0.431)
            (4000, 0.429)
            (5000, 0.454)
            (6000, 0.467)
            (7000, 0.445)
            (8000, 0.432)
            (9000, 0.424)
            (10000, 0.387)
        };
        
        \addplot+[]  
        coordinates {
            (1000, 1.127)
            (2000, 1.076)
            (3000, 1.369)
            (4000, 1.709)
            (5000, 1.763)
            (6000, 1.763)
            (7000, 1.709)
            (8000, 2.121)
            (9000, 2.432)
            (10000, 2.582)
        };
        \addplot+[]  
        coordinates {
            (1000, 1.171)
            (2000, 1.085)
            (3000, 1.314)
            (4000, 1.706)
            (5000, 1.763)
            (6000, 1.768)
            (7000, 1.709)
            (8000, 1.983)
            (9000, 2.392)
            (10000, 2.469)
        };
        \addplot+[]  
        coordinates {
            (1000, 1.040)
            (2000, 0.990)
            (3000, 1.198)
            (4000, 1.626)
            (5000, 1.709)
            (6000, 1.713)
            (7000, 1.701)
            (8000, 1.770)
            (9000, 2.180)
            (10000, 2.262)
        };

        \legend{
            {1m, \scalebox{0.85}{$ \testsignificance $}=0.01},
            {100k, \scalebox{0.85}{$ \testsignificance $}=0.001},
            {100k, \scalebox{0.85}{$ \testsignificance $}=0.01},
            {100k, \scalebox{0.85}{$ \testsignificance $}=0.1},
        }
    \end{axis}
\end{tikzpicture}}
    \hfill
    \subfloat{\begin{tikzpicture}[trim axis right]
    \begin{semilogyaxis}[
        parameter sensitivity plots,
        title = {\textbf{(b) Synopsis size (MB)}},
        ytick = {0.01, 0.1, 1, 10, 100, 1000},
        ymin=0.04, ymax=200,
        log origin=infty,  
        log ticks with fixed point,
        legend pos=north east,
        ]
        
        \addplot+[]  
        coordinates {
            (1000, 35.73)
            (2000, 24.53)
            (3000, 14.89)
            (4000, 10.44)
            (5000, 7.43)
            (6000, 5.32)
            (7000, 4.66)
            (8000, 4.06)
            (9000, 3.57)
            (10000, 2.91)
        };
        
        \addplot+[]  
        coordinates {
            (1000, 1.45)
            (2000, 0.65)
            (3000, 0.43)
            (4000, 0.29)
            (5000, 0.24)
            (6000, 0.21)
            (7000, 0.18)
            (8000, 0.14)
            (9000, 0.13)
            (10000, 0.12)
        };

        \legend{
            {1m, \scalebox{0.85}{$ \testsignificance $}=0.01},
            {100k, \scalebox{0.85}{$ \testsignificance $}=0.01}
        }

    \end{semilogyaxis}
\end{tikzpicture}}
    \vspace{-0.4em}
    \caption{
        \histname{} performance on the scaled-up \textit{Flights} dataset for different parameter sets.
    }
    \label{fig:parameter_sensitivity}
\end{figure}
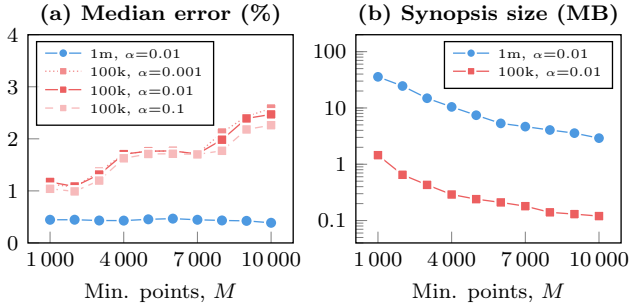

In general, higher $ \nsamples $, higher $ \testsignificance $ and lower $ \minpts $ all correspond to higher accuracy at the cost of larger synopsis size and longer construction time.
However, $ \nsamples $ has greatest impact on performance, while $ \testsignificance $ has minimal impact.
This is illustrated in \autoref{fig:parameter_sensitivity} for the scaled-up \textit{Flights} dataset, where $ \testsignificance $ has near-zero impact except on accuracy for $ \nsamples $ = 100k in some cases.
In our tests, we have also observed that construction time scales linearly with $ \nsamples $ and query latency is largely consistent across different parameter sets.


\subsection{Accuracy}
\label{sec:evaluation:query_error}

\begin{table*}[!t]
    \centering
    \caption{
        Median relative error (\%).
    }
    \hspace{\tablecaptiongapskrinkage}
    \footnotesize
    \begin{tabular}{lcccccc}
    \toprule
                & \multicolumn{3}{c}{\textbf{\textit{Power} dataset}} & \multicolumn{3}{c}{\textbf{\textit{Flights} dataset}} \\
    Aggregation &           PH           &    DeepDB     &  DBEst++   &      PH       &    DeepDB     &        DBEst++        \\ \midrule
    COUNT       &          0.19          & \textbf{0.05} &   24.82    & \textbf{0.38} &     0.41      &         21.65         \\
    SUM         &     \textbf{0.32}      &     14.18     &   56.46    & \textbf{1.15} &     1.72      &         3.55          \\
    AVG         &     \textbf{0.42}      &     0.50      &   17.86    &     0.39      & \textbf{0.28} &         16.95         \\
    VAR         &     \textbf{0.84}      &       -       &   98.50    & \textbf{1.67} &       -       &        100.00         \\
    MIN         &     \textbf{0.00}      &       -       &     -      & \textbf{0.00} &       -       &           -           \\
    MAX         &     \textbf{1.25}      &       -       &     -      & \textbf{4.41} &       -       &           -           \\
    MEDIAN      &     \textbf{0.00}      &       -       &     -      & \textbf{0.29} &       -       &           -           \\ \midrule
    Overall     & \textbf{\textbf{0.20}} &     0.45      &   56.46    & \textbf{0.43} &     0.64      &         28.42         \\ \bottomrule
    \multicolumn{7}{l}{\tiny PH = \histname{}: 1 million samples, DeepDB: 1 million samples, DBEst++: 100k samples.}
\end{tabular}

    \label{tab:error}
\end{table*}

\begin{figure*}[!t]
    \centering
    \subfloat{\adjustbox{valign=t, margin=-0.6em}{\begin{tikzpicture}
    \begin{semilogxaxis}[
        cdf plots,
        title = {\textbf{(b) DeepDB ($ n=146 $)}},
        ]
        \foreach \i in {2,1,3,4} {
            \addplot+ [
            sharp plot, no marks, very thick
            ] table[x index=0, y index=\i] {\datacdfdeepdb};
        }

        \draw[black!50, densely dotted] (axis cs:10,0) -- (axis cs:10,82.5);
        \node[circle, inner sep=0.15em, draw, white, fill=colour4, line width=0.5pt] (dot) at (axis cs:10,82.5) {};
        \node[anchor=south east, inner sep=0.0em] () at (dot.north west) {\scriptsize 82.5\%};

        \legend{\histname{} 1m, \histname{} 100k, DeepDB 1m, DeepDB 100k}
    \end{semilogxaxis}
\end{tikzpicture}}}
    \hfill
    \subfloat{\adjustbox{valign=t, margin=-0.6em}{\begin{tikzpicture}
    \begin{semilogxaxis}[
        cdf plots,
        title = {\textbf{(a) DBEst++  ($ n=86 $)}},
        cycle list={
            {colour3!80, thick},
            {colour3!80, thick, densely dashed},
            {colour2!80, thick},
            {colour2!80, thick, densely dashed},
        },
        ylabel={},
        ]
        \foreach \i in {2,1,3,4} {
            \addplot+ [
            sharp plot, no marks, very thick
            ] table[x index=0, y index=\i] {\datacdfdbestpp};
        }

        \draw[black!50, densely dotted] (axis cs:10,0) -- (axis cs:10,95.1);
        \node[circle, inner sep=0.15em, draw, white, fill=colour4, line width=0.5pt] (dot) at (axis cs:10,95.1) {};
        \node[anchor=north west] () at (dot.south west) {\scriptsize 95.1\%};

        \legend{\histname{} 1m, \histname{} 100k, DBEst++ 100k, DBEst++ 10k}
    \end{semilogxaxis}
\end{tikzpicture}}}
    \hfill
    \subfloat{\adjustbox{valign=t, margin=-0.6em}{\begin{tikzpicture}
    \begin{semilogxaxis}[
        cdf plots,
        title = {\textbf{(c) All queries ($ n=872 $)}},
        ylabel={},
        ]
        \foreach \i in {2,1} {
            \addplot+ [
            sharp plot, no marks, very thick
            ] table[x index=0, y index=\i] {\datacdfall};
        }

        \draw[black!50, densely dotted] (axis cs:10,0) -- (axis cs:10,85.1);
        \node[circle, inner sep=0.15em, draw, white, fill=colour4, line width=0.5pt] (dot) at (axis cs:10,85.1) {};
        \node[anchor=north west, inner sep=0.3em] () at (dot.south) {\scriptsize 85.1\%};

        \legend{\histname{} 1m, \histname{} 100k}
    \end{semilogxaxis}
\end{tikzpicture}}}
    \hfill
    \subfloat{\adjustbox{valign=t, margin=-0.6em}{\newcommand{\figrealcompbarwidth}{7pt}

\begin{tikzpicture}[trim axis right]
    \begin{axis}[
        my bar plot,
        width = 0.21*\textwidth, height = 0.16*\textwidth,
        enlarge x limits=0.2,
        bar width = \figrealcompbarwidth,
        ylabel = {Median error (\%)},
        ytick = {0,1,2,3,4},
        xtick=data,
        xtick = {1,2,3,4},
        xticklabels = {Power, Flights, Power, Flights},
        title = {\textbf{(d) Real vs. IDEBench}$ ^* $},
        legend image post style={scale=0.8},
        name = graph,
        ]

        \addplot+[bar shift=-0.5*\figrealcompbarwidth] coordinates {(1, 0.125914) (2, 0.660598)};
        \addplot+[bar shift=0.5*\figrealcompbarwidth] coordinates {(1, 0.132675) (2, 0.345782)};

        \addplot+[bar shift=-0.5*\figrealcompbarwidth] coordinates {(3, 4.186534) (4, 3.370849)};
        \addplot+[bar shift=0.5*\figrealcompbarwidth] coordinates {(3, 0.737420) (4, 0.752934)};

        \legend{\histname{} Real, \histname{} IDEBench, DeepDB Real, DeepDB IDEBench}
    \end{axis}

    \node[anchor=north east, yshift=-2.1em, inner sep=0] () at (graph.south east) {\tiny $ ^* $1m sample size used for all models};
\end{tikzpicture}}}
    \caption{
        (a)--(c) CDF plots of query error for different subsets of queries across both datasets, (d) query performance on the original real-world data and IDEBench-generated data of the same size.
    }
    \label{fig:error}
\end{figure*}
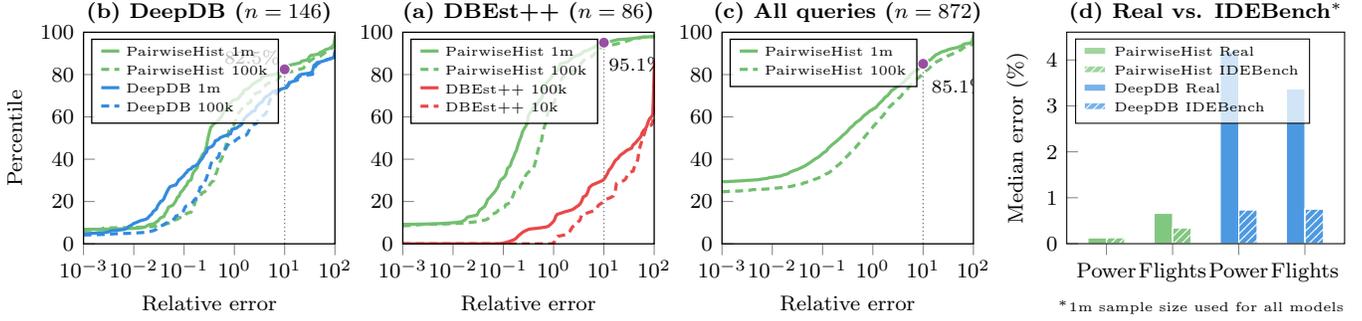

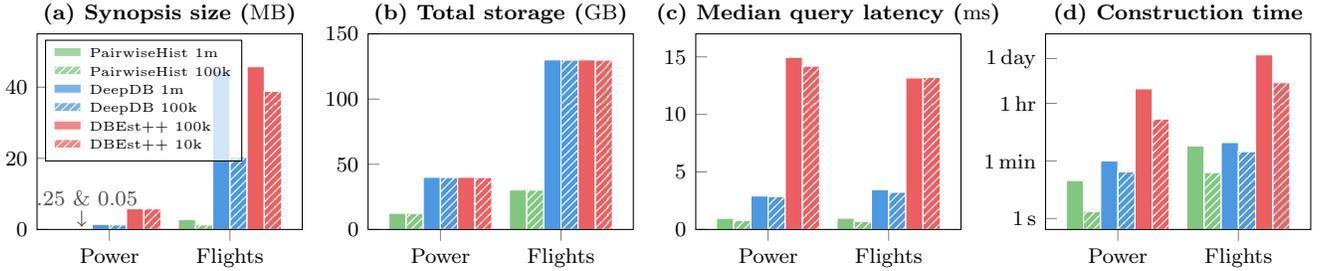
\begin{figure*}[!t]
    \centering
    \subfloat{\begin{tikzpicture}[trim axis right]
    \begin{axis}[
        evaluation bar plots,
        title = {\textbf{(a) Synopsis size (\unit{MB})}},
        ymax=55,
        ]

        \addplot+[] coordinates {(1, 0.25) (2, 2.76)};
        \addplot+[] coordinates {(1, 0.05) (2, 1.31)};
        \addplot+[] coordinates {(1, 1.43) (2, 44.52)};
        \addplot+[] coordinates {(1, 1.31) (2, 20.42)};
        \addplot+[] coordinates {(1, 5.81) (2, 45.76)};
        \addplot+[] coordinates {(1, 5.81) (2, 38.90)};

        \node[anchor=center, inner sep=0.2em, black!70] (annotation1) at (axis cs:0.77,9) {0.25 \& 0.05};
        \draw[->, black!70] (annotation1.south) -- (axis cs:0.77,0.8);

        \legend{\histname{} 1m, \histname{} 100k, DeepDB 1m, DeepDB 100k, DBEst++ 100k, DBEst++ 10k}
    \end{axis}
\end{tikzpicture}}
    \hfill
    \subfloat{\begin{tikzpicture}[trim axis right]
    \begin{axis}[
        evaluation bar plots,
        title = {\textbf{(b) Total storage (\unit{GB})}},
        ymax=150,
        ]

        \addplot+[] coordinates {(1, 12.411518) (2, 30.267561)};
        \addplot+[] coordinates {(1, 12.411314) (2, 30.264094)};
        \addplot+[] coordinates {(1, 40.001427) (2, 130.044515)};
        \addplot+[] coordinates {(1, 40.001314) (2, 130.020416)};
        \addplot+[] coordinates {(1, 40.005815) (2, 130.045759)};
        \addplot+[] coordinates {(1, 40.005806) (2, 130.038904)};

    \end{axis}
\end{tikzpicture}}
    \hfill
    \subfloat{\begin{tikzpicture}[trim axis right]
    \begin{axis}[
        evaluation bar plots,
        title = {\textbf{(c) Median query latency (\unit{ms})}},
        ymax=17,
        ]

        \addplot+[] coordinates {(1, 0.97) (2, 0.98)};  
        \addplot+[] coordinates {(1, 0.78) (2, 0.71)};  
        \addplot+[] coordinates {(1, 2.91) (2, 3.45)};  
        \addplot+[] coordinates {(1, 2.87) (2, 3.25)};  
        \addplot+[] coordinates {(1, 14.96) (2, 13.17)};  
        \addplot+[] coordinates {(1, 14.21) (2, 13.21)};  
    \end{axis}
\end{tikzpicture}}
    \hfill
    \subfloat{\begin{tikzpicture}[trim axis right]
    \begin{semilogyaxis}[
        evaluation bar plots,
        title = {\textbf{(d) Construction time}},
        ymax = 500000,
        ytick = {1, 60, 3600, 24*3600},
        yticklabels = {\qty{1}{s}, \qty{1}{min}, \qty{1}{hr}, \qty{1}{day}}
        ]

        \addplot+[] coordinates {(1, 14.819) (2, 173.402)};
        \addplot+[] coordinates {(1, 1.622) (2, 26.614)};
        \addplot+[] coordinates {(1, 59.588) (2, 220.581)};
        \addplot+[] coordinates {(1, 27.569) (2, 117.197)};
        \addplot+[] coordinates {(1, 10043.349) (2, 112343.206)};
        \addplot+[] coordinates {(1, 1194.751) (2, 15673.884)};

    \end{semilogyaxis}
\end{tikzpicture}}
    \vspace{-0.5em}
    \caption{
        Storage and runtime performance comparison on the scaled-up datasets.
    }
    \label{fig:evaluation_other}
\end{figure*}

\textit{Query Accuracy.}
\autoref{tab:error} presents the median error by dataset and aggregation function for the scaled-up experiments.
As can be seen, \histname{} (denoted PH) performs well across all aggregation functions and delivers between 1.5--2.3$ \times $ better overall accuracy than DeepDB with median errors of just 0.20\% and 0.43\% compared to 0.45\% and 0.64\% for the two datasets.
Note that a smaller sample size was used for DBEst++ were due to its prohibitively long training time (see \autoref{sec:evaluation:construction}).

We also compared the distribution of query errors over the subset of queries supported by DeepDB (\autoref{fig:error}(a)) and DBEst++ (\autoref{fig:error}(b)) as CDF plots, while the error distribution for \histname{} over all queries is shown in (\autoref{fig:error}(c)).
As can be seen, \histname{} provides a better error distribution in each case.
Indeed, with just 100k samples, \histname{} achieves a higher probability of sub-1\% error than DeepDB with 1 million samples.
Overall, 85.1\% of queries have sub-10\% error with \histname{}, as highlighted in \autoref{fig:error}(c).

During our evaluation, DeepDB was observed to perform significantly worse on real-world data compared to IDEBench-generated synthetic data.
To demonstrate this, we generated synthetic versions of the \textit{Power} and \textit{Flights} datasets using IDEBench with the same number of rows as the original data and tested identical queries on them using DeepDB and \histname{}.
The resulting median errors are shown in \autoref{fig:error}(d).
As can be seen, DeepDB performs far worse on the real data compared to the IDEBench-generated data.
IDEBench generates synthetic data by applying normalisation and Gaussian models.
This suggests that, while DeepDB performs well on standard test data, it may not perform as effectively in the real world, where data is less well-behaved.
\histname{}, on the other hand, performs consistently well and has up to 31$ \times $ lower error than DeepDB on the real datasets.

\begin{table*}[!t]
    \centering
    \caption{Bounds accuracy rate and width.}
    \hspace{\tablecaptiongapskrinkage}
    \footnotesize
    \begin{tabular}{lcccc}
    \toprule
                        & \multicolumn{2}{c}{\textbf{Correct rate (\%)}} & \multicolumn{2}{c}{\textbf{Width (\%)}} \\
    Dataset             &  \histname{}  &             DeepDB             & \histname{} &          DeepDB           \\ \midrule
    Power (original)    & \textbf{70.0} &              40.0              &     4.4     &       \textbf{0.7}        \\
    Power (1 billion)   & \textbf{80.0} &              51.2              &     3.4     &       \textbf{0.6}        \\
    Flights (original)  & \textbf{78.8} &              50.0              &     8.7     &       \textbf{3.0}        \\
    Flights (1 billion) & \textbf{78.8} &              75.8              &     4.3     &       \textbf{2.3}        \\ \bottomrule
\end{tabular}

    \label{tab:bounds}
\end{table*}

\textit{Query Bounds.}
Query bounds should be both accurate (i.e., contain the true result) and narrow.
\autoref{tab:bounds} lists the percentage of queries for which DeepDB and \histname{} provide accurate bounds and the median bound widths as a percentage of the exact result for the subset of queries supported by DeepDB (DBEst++ does not provide bounds).
A significance value of 0.99 was used for DeepDB.
As shown, \histname{} provides more accurate bounds than DeepDB, especially for the real-world datasets.
DeepDB has consistently narrower bounds, but given their lower accuracy, this may indicate that its bounds are overly optimistic.

\subsection{Storage requirements}
\label{sec:evaluation:storage}

Synopsis sizes are shown in \autoref{fig:evaluation_other}(a).
As can be seen, \histname{} requires the least storage in all cases and is at least 11$ \times $ smaller than DeepDB and DBEst++ (\qty{0.25}{MB} vs. \qty{2.75}{MB} for \textit{Power} with 1m samples).
Additionally, with \histname{} built directly on compressed data, total storage requirements are reduced, since data can be permanently stored in compressed format.
Total storage requirements are shown in \autoref{fig:evaluation_other}(b), where \histname{} delivers savings of 3.2--4.3$ \times $.

\subsection{Query latency}
\label{sec:evaluation:latency}

Median query latency, including the time to calculate query bounds (for \histname{} and DeepDB), is shown in \autoref{fig:evaluation_other}(c).
As can be seen, \histname{} is significantly faster than the state-of-the-art approaches with an overall median latency of just \qty{0.94}{ms} (for 1m samples), which is 3.5$ \times $ faster than DeepDB and 15$ \times $ faster than DBEst++.
\histname{}'s low latency can be attributed to most aggregations requiring just a handful of small matrix multiplications.
The corresponding median latency for exact query processing using SQLite, which we used for the ground truth, was \qty{306.8}{seconds}, which makes \histname{} over 300\,000$ \times $ faster.


\subsection{Construction time}
\label{sec:evaluation:construction}

Finally, synopsis construction times are displayed in \autoref{fig:evaluation_other}(d).
As can be seen, \histname{} consistently requires the least time, being 1.2--4$ \times $ faster than DeepDB, while DBEst++ is more than two orders of magnitude slower.
Indeed, for the Flights dataset, DBEst++ requires over \qty{30}{hours} for 100k samples, while \histname{} requires less than \qty{3}{minutes} with 1 million samples.

	\section{Conclusion}
\label{sec:conclustion}

In this paper, we propose a novel AQP technique called \histname{} and a novel framework for AQP that leverages data compression to reduce overall storage requirements.
By using a collection of histograms approach, efficient storage encoding and various query execution optimisations, we are able to simultaneously deliver significant improvements in terms of accuracy, latency, synopsis size and construction time compared to state-of-the-art AQP methods.
In future work, we intend to investigate histogram updates, online refinement and multi-table support.

	\section{Acknowledgements}

This work is supported by the Analytics Straight on Compressed IoT Data (Light-IoT) project (Grant No. 0136-00376B), granted by the Danish Council for Independent Research, and Aarhus University's DIGIT Centre.
Computation was partially performed on the UCloud system, which is managed by the eScience Center at the University of Southern Denmark.

    \bibliographystyle{IEEEtranN}
	\bibliography{library}

\begin{thebibliography}{62}
\providecommand{\natexlab}[1]{#1}
\providecommand{\url}[1]{#1}
\csname url@samestyle\endcsname
\providecommand{\newblock}{\relax}
\providecommand{\bibinfo}[2]{#2}
\providecommand{\BIBentrySTDinterwordspacing}{\spaceskip=0pt\relax}
\providecommand{\BIBentryALTinterwordstretchfactor}{4}
\providecommand{\BIBentryALTinterwordspacing}{\spaceskip=\fontdimen2\font plus
\BIBentryALTinterwordstretchfactor\fontdimen3\font minus
  \fontdimen4\font\relax}
\providecommand{\BIBforeignlanguage}[2]{{%
\expandafter\ifx\csname l@#1\endcsname\relax
\typeout{** WARNING: IEEEtranN.bst: No hyphenation pattern has been}%
\typeout{** loaded for the language `#1'. Using the pattern for}%
\typeout{** the default language instead.}%
\else
\language=\csname l@#1\endcsname
\fi
#2}}
\providecommand{\BIBdecl}{\relax}
\BIBdecl

\bibitem[Li and Li(2018)]{Li_2018}
K.~Li and G.~Li, ``Approximate query processing: What is new and where to go?''
  \emph{Data Science and Engineering}, vol.~3, no.~4, pp. 379--397, 2018.

\bibitem[Akash et~al.(2022)Akash, Lai, and Lin]{Akash_2022}
P.~S. Akash, W.-C. Lai, and P.-W. Lin, ``Online aggregation based approximate
  query processing: A literature survey,'' 2022.

\bibitem[Liang et~al.(2021)Liang, Sintos, Shang, and Krishnan]{Liang_2021}
X.~Liang, S.~Sintos, Z.~Shang, and S.~Krishnan, ``Combining aggregation and
  sampling (nearly) optimally for approximate query processing,'' in
  \emph{Proceedings of the {ACM} {SIGMOD} International Conference on
  Management of Data}, 2021.

\bibitem[Peng et~al.(2018)Peng, Zhang, Wang, and Pei]{Peng_2018}
J.~Peng, D.~Zhang, J.~Wang, and J.~Pei, ``{AQP++}: Connecting approximate query
  processing with aggregate precomputation for interactive analytics,'' in
  \emph{Proceedings of the {ACM} {SIGMOD} International Conference on
  Management of Data}, 2018, pp. 1477--1492.

\bibitem[Park et~al.(2018)Park, Mozafari, Sorenson, and Wang]{Park_2018}
Y.~Park, B.~Mozafari, J.~Sorenson, and J.~Wang, ``{VerdictDB}: Universalizing
  approximate query processing,'' in \emph{Proceedings of the {ACM} {SIGMOD}
  International Conference on Management of Data}, 2018.

\bibitem[Hurst et~al.(2021)Hurst, Zhang, Lucani, and Assent]{Hurst_2021}
A.~Hurst, Q.~Zhang, D.~E. Lucani, and I.~Assent, ``Direct analytics of
  generalized deduplication compressed {IoT} data,'' in \emph{{IEEE} Global
  Communications Conference ({GLOBECOM})}, 2021.

\bibitem[Hurst et~al.(2022)Hurst, Lucani, Assent, and Zhang]{Hurst_2022}
A.~Hurst, D.~E. Lucani, I.~Assent, and Q.~Zhang, ``{GLEAN}: Generalized
  deduplication enabled approximate edge analytics,'' \emph{IEEE Internet of
  Things Journal}, 2022.

\bibitem[Hurst et~al.(2024)Hurst, Lucani, and Zhang]{Hurst_2024}
A.~Hurst, D.~E. Lucani, and Q.~Zhang, ``{GreedyGD}: Enhanced generalized
  deduplication for direct analytics in {IoT},'' \emph{IEEE Transactions on
  Industrial Informatics}, 2024.

\bibitem[Vestergaard et~al.(2019{\natexlab{a}})Vestergaard, Zhang, and
  Lucani]{Vestergaard_2019a}
R.~Vestergaard, Q.~Zhang, and D.~E. Lucani, ``Generalized deduplication:
  Lossless compression for large amounts of small {IoT} data,'' in
  \emph{European Wireless}, 2019.

\bibitem[Vestergaard et~al.(2019{\natexlab{b}})Vestergaard, Zhang, and
  Lucani]{Vestergaard_2019b}
------, ``Generalized deduplication: Bounds, convergence, and asymptotic
  properties,'' in \emph{{IEEE} Global Communications Conference
  ({GLOBECOM})}.\hskip 1em plus 0.5em minus 0.4em\relax {IEEE}, 2019.

\bibitem[Vestergaard et~al.(2020)Vestergaard, Lucani, and
  Zhang]{Vestergaard_2020}
R.~Vestergaard, D.~E. Lucani, and Q.~Zhang, ``A randomly accessible lossless
  compression scheme for time-series data,'' in \emph{{IEEE} {INFOCOM}}, 2020.

\bibitem[Vestergaard et~al.(2021)Vestergaard, Zhang, Sipos, and
  Lucani]{Vestergaard_2021}
R.~Vestergaard, Q.~Zhang, M.~Sipos, and D.~E. Lucani, ``Titchy: Online
  time-series compression with random access for the internet of things,''
  \emph{IEEE Internet of Things Journal}, vol.~8, no.~24, pp. 17\,568--17\,583,
  2021.

\bibitem[Sehat et~al.(2022)Sehat, Kloborg, Mørup, Pagnin, and
  Lucani]{Sehat_2022}
H.~Sehat, A.~L. Kloborg, C.~Mørup, E.~Pagnin, and D.~E. Lucani, ``Bonsai: A
  generalized look at dual deduplication,'' 2022.

\bibitem[Nielsen et~al.(2019)Nielsen, Vestergaard, Yazdani, Talasila, Lucani,
  and Sipos]{Nielsen_2019}
L.~Nielsen, R.~Vestergaard, N.~Yazdani, P.~Talasila, D.~E. Lucani, and
  M.~Sipos, ``Alexandria: A proof-of-concept implementation and evaluation of
  generalised data deduplication,'' in \emph{{IEEE} ({GLOBECOM}) Workshops},
  2019.

\bibitem[Fehér et~al.(2022)Fehér, Lucani, and Chatzigeorgiou]{Feher_2022}
M.~Fehér, D.~E. Lucani, and I.~Chatzigeorgiou, ``An adaptive column
  compression family for self-driving databases,'' 2022, {arXiv:2209.02334v1}.

\bibitem[Göttel et~al.(2020)Göttel, Nielsen, Yazdani, Felber, Lucani, and
  Schiavoni]{Goettel_2020}
C.~Göttel, L.~Nielsen, N.~Yazdani, P.~Felber, D.~E. Lucani, and V.~Schiavoni,
  ``Hermes: enabling energy-efficient {IoT} networks with generalized
  deduplication,'' in \emph{Proceedings of the 14th {ACM} International
  Conference on Distributed and Event-based Systems}.\hskip 1em plus 0.5em
  minus 0.4em\relax {ACM}, 2020.

\bibitem[Zhang et~al.(2014)Zhang, Yan, Chen, Wang, Moscibroda, and
  Zhang]{Zhang_2014}
J.~Zhang, Y.~Yan, L.~J. Chen, M.~Wang, T.~Moscibroda, and Z.~Zhang,
  ``Impression store: Compressive sensing-based storage for big data
  analytics,'' in \emph{6th USENIX Workshop on Hot Topics in Cloud Computing
  (HotCloud 14)}.\hskip 1em plus 0.5em minus 0.4em\relax {USENIX Association},
  2014.

\bibitem[Kong et~al.(2022)Kong, Tan, Huang, Chen, Wang, Jin, Zeng, Khan, and
  Das]{Kong_2022}
L.~Kong, J.~Tan, J.~Huang, G.~Chen, S.~Wang, X.~Jin, P.~Zeng, M.~Khan, and
  S.~K. Das, ``Edge-computing-driven {Internet of Things}: A survey,''
  \emph{{ACM} Computing Surveys}, vol.~55, no.~8, pp. 1--41, 2022.

\bibitem[Nayak et~al.(2022)Nayak, Patgiri, Waikhom, and Ahmed]{Nayak_2022}
S.~Nayak, R.~Patgiri, L.~Waikhom, and A.~Ahmed, ``A review on {Edge} analytics:
  Issues, challenges, opportunities, promises, future directions, and
  applications,'' \emph{Digital Communications and Networks}, 2022.

\bibitem[Hilprecht et~al.(2020)Hilprecht, Schmidt, Kulessa, Molina, Kersting,
  and Binnig]{Hilprecht_2020}
B.~Hilprecht, A.~Schmidt, M.~Kulessa, A.~Molina, K.~Kersting, and C.~Binnig,
  ``{DeepDB}: Learn from data, not from queries,'' in \emph{Proceedings of the
  VLDB Endowment}, vol.~13, no.~7, 2020, pp. 992--1005.

\bibitem[Ma et~al.(2021)Ma, Shanghooshabad, Almasi, Kurmanji, and
  Triantafillou]{Ma_2021}
Q.~Ma, A.~M. Shanghooshabad, M.~Almasi, M.~Kurmanji, and P.~Triantafillou,
  ``Learned approximate query processing: Make it light, accurate and fast,''
  in \emph{Conference on Innovative Data Systems Research}, 2021.

\bibitem[Eichmann et~al.(2020)Eichmann, Zgraggen, Binnig, and
  Kraska]{Eichmann_2020}
P.~Eichmann, E.~Zgraggen, C.~Binnig, and T.~Kraska, ``{IDEBench}: A benchmark
  for interactive data exploration,'' in \emph{Proceedings of the {ACM}
  {SIGMOD} International Conference on Management of Data}.\hskip 1em plus
  0.5em minus 0.4em\relax {ACM}, 2020.

\bibitem[Chaudhuri et~al.(2017)Chaudhuri, Ding, and Kandula]{Chaudhuri_2017}
S.~Chaudhuri, B.~Ding, and S.~Kandula, ``Approximate query processing: No
  silver bullet,'' in \emph{ACM International Conference on Management of
  Data}, ser. SIGMOD/PODS’17.\hskip 1em plus 0.5em minus 0.4em\relax {ACM},
  2017.

\bibitem[Ahmadvand et~al.(2019)Ahmadvand, Goudarzi, and
  Foroutan]{Ahmadvand_2019}
H.~Ahmadvand, M.~Goudarzi, and F.~Foroutan, ``Gapprox: using gallup approach
  for approximation in big data processing,'' \emph{Journal of Big Data},
  vol.~6, no.~1, 2019.

\bibitem[Babcock et~al.(2003)Babcock, Chaudhuri, and Das]{Babcock_2003}
B.~Babcock, S.~Chaudhuri, and G.~Das, ``Dynamic sample selection for
  approximate query processing,'' in \emph{Proceedings of the {ACM} {SIGMOD}
  International Conference on Management of Data}, 2003, pp. 539--550.

\bibitem[Agarwal et~al.(2013)Agarwal, Mozafari, Panda, Milner, Madden, and
  Stoica]{Agarwal_2013}
S.~Agarwal, B.~Mozafari, A.~Panda, H.~Milner, S.~Madden, and I.~Stoica,
  ``{BlinkDB},'' in \emph{Proceedings of the 8th {ACM} European Conference on
  Computer Systems}, 2013.

\bibitem[Sanca et~al.(2023)Sanca, Chrysogelos, and Ailamaki]{Sanca_2023}
V.~Sanca, P.~Chrysogelos, and A.~Ailamaki, ``Laqy: Efficient and reusable query
  approximations via lazy sampling,'' in \emph{Proceedings of the {ACM}
  {SIGMOD} International Conference on Management of Data}, 2023.

\bibitem[Ioannidis and Poosala(1995)]{Ioannidis_1995}
Y.~E. Ioannidis and V.~Poosala, ``Balancing histogram optimality and
  practicality for query result size estimation,'' \emph{{ACM} {SIGMOD}
  Record}, vol.~24, no.~2, pp. 233--244, 1995.

\bibitem[To et~al.(2013)To, Chiang, and Shahabi]{To_2013}
H.~To, K.~Chiang, and C.~Shahabi, ``Entropy-based histograms for selectivity
  estimation,'' in \emph{Proceedings of the 22nd {ACM} International Conference
  on Conference on Information {\&} Knowledge Management ({CIKM})}.\hskip 1em
  plus 0.5em minus 0.4em\relax {ACM} Press, 2013.

\bibitem[Acharya et~al.(2015)Acharya, Diakonikolas, Hegde, Li, and
  Schmidt]{Acharya_2015}
J.~Acharya, I.~Diakonikolas, C.~Hegde, J.~Z. Li, and L.~Schmidt, ``Fast and
  near-optimal algorithms for approximating distributions by histograms,'' in
  \emph{Proceedings of the 34th {ACM} {SIGMOD}-{SIGACT}-{SIGAI} Symposium on
  Principles of Database Systems}, 2015.

\bibitem[Diakonikolas et~al.(2018)Diakonikolas, Li, and
  Schmidt]{Diakonikolas_2018}
I.~Diakonikolas, J.~Li, and L.~Schmidt, ``Fast and sample near-optimal
  algorithms for learning multidimensional histograms,'' in \emph{Proceedings
  of the 31st Conference On Learning Theory}, vol.~75, 2018, pp. 819--842.

\bibitem[Singla and Eldawy(2022)]{Singla_2022}
S.~Singla and A.~Eldawy, \emph{Flexible Computation of Multidimensional
  Histograms}, New York, NY, USA, 2022, ch.~13, pp. 119--130.

\bibitem[Cormode et~al.(2011)Cormode, Garofalakis, Haas, and
  Jermaine]{Cormode_2011}
G.~Cormode, M.~Garofalakis, P.~J. Haas, and C.~Jermaine, ``Synopses for massive
  data: Samples, histograms, wavelets, sketches,'' in \emph{Foundations and
  Trends in Databases}, vol.~4, no. 1-3, 2011, pp. 1--294.

\bibitem[Zhang and Wang(2021{\natexlab{a}})]{Zhang_2021a}
M.~Zhang and H.~Wang, ``Selectivity estimation with density-model-based
  multidimensional histogram,'' \emph{Knowledge and Information Systems},
  vol.~63, no.~4, pp. 971--992, 2021.

\bibitem[Shekelyan et~al.(2017)Shekelyan, Dign\"{o}s, and
  Gamper]{Shekelyan_2017}
M.~Shekelyan, A.~Dign\"{o}s, and J.~Gamper, ``{DigitHist}: A histogram-based
  data summary with tight error bounds,'' in \emph{Proceedings of the VLDB
  Endowment}, vol.~10, no.~11, 2017, pp. 1514--1525.

\bibitem[Bruno et~al.(2001)Bruno, Chaudhuri, and Gravano]{Bruno_2001}
N.~Bruno, S.~Chaudhuri, and L.~Gravano, ``{STHoles},'' in \emph{Proceedings of
  the {ACM} {SIGMOD} International Conference on Management of Data}, 2001.

\bibitem[Lee et~al.(2022)Lee, Nam, Park, and Kim]{Lee_2022}
T.~Lee, K.~Nam, C.~S. Park, and S.-S. Kim, ``Exploiting machine learning models
  for approximate query processing,'' in \emph{{IEEE} International Conference
  on Big Data}, 2022.

\bibitem[Kulessa et~al.(2018)Kulessa, Molina, Binnig, Hilprecht, and
  Kersting]{Kulessa_2018}
M.~Kulessa, A.~Molina, C.~Binnig, B.~Hilprecht, and K.~Kersting, ``Model-based
  approximate query processing,'' 2018.

\bibitem[Kulessa et~al.(2019)Kulessa, Hilprecht, Molina, Engineering, Data,
  Lab, Learning, and Lab]{Kulessa_2019}
M.~Kulessa, B.~Hilprecht, A.~Molina, K.~Engineering, G.~Data, M.~Lab,
  M.~Learning, and Lab, ``Towards model-based approximate query processing,''
  in \emph{1st International Workshop on Applied AI for Database Systems and
  Applications}, 2019.

\bibitem[Ma and Triantafillou(2019)]{Ma_2019}
Q.~Ma and P.~Triantafillou, ``{DBEst}: Revisiting approximate query processing
  engines with machine learning models,'' in \emph{Proceedings of the {ACM}
  {SIGMOD} International Conference on Management of Data}, 2019, pp.
  1553--1570.

\bibitem[Zhang and Wang(2021{\natexlab{b}})]{Zhang_2021}
M.~Zhang and H.~Wang, ``{LAQP}: Learning-based approximate query processing,''
  \emph{Information Sciences}, vol. 546, pp. 1113--1134, 2021.

\bibitem[Sheoran et~al.(2022)Sheoran, Mitra, Porwal, Ghetia, Varshney, Mai,
  Rao, and Maddukuri]{Sheoran_2022}
N.~Sheoran, S.~Mitra, V.~Porwal, S.~Ghetia, J.~Varshney, T.~Mai, A.~Rao, and
  V.~Maddukuri, ``\BIBforeignlanguage{EN-US}{Conditional generative model based
  predicate-aware query approximation},''
  \emph{\BIBforeignlanguage{EN-US}{Proceedings of the {AAAI} Conference on
  Artificial Intelligence}}, vol.~36, no.~8, pp. 8259--8266, 2022.

\bibitem[Zeighami et~al.(2022)Zeighami, Shahabi, and Sharan]{Zeighami_2023}
S.~Zeighami, C.~Shahabi, and V.~Sharan, ``Neurosketch: Fast and approximate
  evaluation of range aggregate queries with neural networks,''
  \emph{Proceedings of the ACM on Management of Data}, vol.~1, no.~1, pp.
  1--26, 2022.

\bibitem[Fallahian et~al.(2022)Fallahian, Dorodchi, and Kreth]{Fallahian_2022}
M.~Fallahian, M.~Dorodchi, and K.~Kreth, ``Gan-based tabular data generator for
  constructing synopsis in approximate query processing: Challenges and
  solutions,'' 2022.

\bibitem[Müller et~al.(2018)Müller, Moerkotte, and Kolb]{M_ller_2018}
M.~Müller, G.~Moerkotte, and O.~Kolb, ``Improved selectivity estimation by
  combining knowledge from sampling and synopses,'' in \emph{Proceedings of the
  {VLDB} Endowment}, vol.~11, no.~9, 2018, pp. 1016--1028.

\bibitem[Zeng et~al.(2014)Zeng, Gao, Mozafari, and Zaniolo]{Zeng_2014}
K.~Zeng, S.~Gao, B.~Mozafari, and C.~Zaniolo, ``The analytical bootstrap: A new
  method for fast error estimation in approximate query processing,'' in
  \emph{Proceedings of the {ACM} {SIGMOD} International Conference on
  Management of Data}, 2014, pp. 277--288.

\bibitem[Li et~al.(2019)Li, Zhang, Li, Tao, and Yan]{Li_2019}
K.~Li, Y.~Zhang, G.~Li, W.~Tao, and Y.~Yan, ``Bounded approximate query
  processing,'' \emph{{IEEE} Transactions on Knowledge and Data Engineering},
  vol.~31, no.~12, pp. 2262--2276, 2019.

\bibitem[Agarwal et~al.(2014)Agarwal, Milner, Kleiner, Talwalkar, Jordan,
  Madden, Mozafari, and Stoica]{Agarwal_2014}
S.~Agarwal, H.~Milner, A.~Kleiner, A.~Talwalkar, M.~Jordan, S.~Madden,
  B.~Mozafari, and I.~Stoica, ``Knowing when you're wrong: Building fast and
  reliable approximate query processing systems,'' in \emph{Proceedings of the
  {ACM} {SIGMOD} International Conference on Management of Data}, 2014, pp.
  481--492.

\bibitem[Scott(2009)]{Scott_2009}
D.~W. Scott, ``Sturges' rule,'' \emph{{WIREs} Computational Statistics},
  vol.~1, no.~3, pp. 303--306, 2009.

\bibitem[Makonin et~al.(2016)Makonin, Ellert, Baji{\'{c}}, and
  Popowich]{Makonin_2016}
\BIBentryALTinterwordspacing
S.~Makonin, B.~Ellert, I.~V. Baji{\'{c}}, and F.~Popowich, ``Electricity,
  water, and natural gas consumption of a residential house in {Canada} from
  2012 to 2014,'' \emph{Scientific Data}, vol.~3, 2016. [Online]. Available:
  \url{http://ampds.org/}
\BIBentrySTDinterwordspacing

\bibitem[Makonin(2016)]{Makonin_2016a}
S.~Makonin, ``{AMPds2}: The almanac of minutely power dataset (version 2),''
  2016.

\bibitem[Huerta et~al.(2016)Huerta, Mosqueiro, Fonollosa, Rulkov, and
  Rodriguez-Lujan]{Huerta_2016}
\BIBentryALTinterwordspacing
R.~Huerta, T.~Mosqueiro, J.~Fonollosa, N.~F. Rulkov, and I.~Rodriguez-Lujan,
  ``Online decorrelation of humidity and temperature in chemical sensors for
  continuous monitoring,'' \emph{Chemometrics and Intelligent Laboratory
  Systems}, vol. 157, pp. 169--176, 2016. [Online]. Available:
  \url{https://archive.ics.uci.edu/ml/datasets/Gas+sensors+for+home+activity+monitoring}
\BIBentrySTDinterwordspacing

\bibitem[Prasad(2020)]{Prasad_2020}
\BIBentryALTinterwordspacing
A.~Prasad, ``Ml prediction for light detection sensor iot,'' 2020. [Online].
  Available:
  \url{https://www.kaggle.com/datasets/aashnaprasad/ml-prediction-for-lightdetection-sensor-iot}
\BIBentrySTDinterwordspacing

\bibitem[Hebrail and Berard(2012)]{Hebrail_2012}
\BIBentryALTinterwordspacing
G.~Hebrail and A.~Berard, ``Individual household electric power consumption
  data set,'' 2012. [Online]. Available:
  \url{https://archive.ics.uci.edu/ml/datasets/Individual+household+electric+power+consumption}
\BIBentrySTDinterwordspacing

\bibitem[Porter(2021)]{Porter_2021}
\BIBentryALTinterwordspacing
M.~Porter, ``Temperature {IoT} on {GCP},'' 2021. [Online]. Available:
  \url{https://www.kaggle.com/datasets/mattpo/temperature-iot-on-gcp}
\BIBentrySTDinterwordspacing

\bibitem[{Udanor Collins} et~al.(2022){Udanor Collins}, {Blessing Ogbuokiri},
  and {Nweke Onyinye}]{Udanor_2022}
\BIBentryALTinterwordspacing
{Udanor Collins}, {Blessing Ogbuokiri}, and {Nweke Onyinye}, ``Sensor based
  aquaponics fish pond datasets,'' 2022. [Online]. Available:
  \url{https://www.kaggle.com/datasets/ogbuokiriblessing/sensor-based-aquaponics-fish-pond-datasets}
\BIBentrySTDinterwordspacing

\bibitem[Hong et~al.(2017)Hong, Gu, and Whitehouse]{Hong_2017}
\BIBentryALTinterwordspacing
D.~Hong, Q.~Gu, and K.~Whitehouse, ``High-dimensional time series clustering
  via cross-predictability,'' in \emph{Proceedings of the 20th International
  Conference on Artificial Intelligence and Statistics}, vol.~54, 2017, pp.
  642--651. [Online]. Available:
  \url{https://www.kaggle.com/datasets/ranakrc/smart-building-system}
\BIBentrySTDinterwordspacing

\bibitem[{USA Department of Transportation}(2016)]{UDT_2016}
\BIBentryALTinterwordspacing
{USA Department of Transportation}, ``2015 flight delays and cancellations,''
  2016. [Online]. Available:
  \url{https://www.kaggle.com/datasets/usdot/flight-delays}
\BIBentrySTDinterwordspacing

\bibitem[{City of Chicago}(2022)]{CityChicago_2022c}
\BIBentryALTinterwordspacing
{City of Chicago}, ``Taxi trips - 2020,'' 2022. [Online]. Available:
  \url{https://data.cityofchicago.org/Transportation/Taxi-Trips-2020/r2u4-wwk3}
\BIBentrySTDinterwordspacing

\bibitem[Garg et~al.(2023)Garg, Mitra, Yu, Gadhia, and Kashettiwar]{Garg_2023}
S.~Garg, S.~Mitra, T.~Yu, Y.~Gadhia, and A.~Kashettiwar, ``Reinforced
  approximate exploratory data analysis,'' \emph{Proceedings of the AAAI
  Conference on Artificial Intelligence}, vol.~37, no.~6, pp. 7660--7669, 2023.

\bibitem[Hilprecht(2020)]{Hilprecht_2020a}
\BIBentryALTinterwordspacing
B.~Hilprecht, ``deepdb-public,'' 2020, {GitHub} repository. Accessed: 1 Apr,
  2023. [Online]. Available:
  \url{https://github.com/DataManagementLab/deepdb-public}
\BIBentrySTDinterwordspacing

\bibitem[Ma(2022)]{Ma_2022}
\BIBentryALTinterwordspacing
Q.~Ma, ``{DBEstClient},'' 2022, {GitHub} repository. Accessed: 1 Apr, 2023.
  [Online]. Available: \url{https://github.com/qingzma/DBEstClient}
\BIBentrySTDinterwordspacing

\end{thebibliography}

\end{document}